\documentclass[transmag]{IEEEtran}
\usepackage{latexsym}
\usepackage{graphicx}
\usepackage{hyperref}
\usepackage[cmex10]{amsmath}
\usepackage{color}
\usepackage{graphicx}
\usepackage{epstopdf}
\usepackage{algorithm}
\usepackage{multicol}
\usepackage{cite}
\usepackage{setspace}
\usepackage[noend]{algpseudocode}
\usepackage[mathscr]{euscript}
\usepackage{multicol}
\usepackage{amsfonts}
\usepackage{xcolor}
\usepackage[utf8]{inputenc}
\usepackage[mathscr]{euscript}
\usepackage{amssymb}
\usepackage{amsfonts}
\usepackage{filecontents}
\usepackage{amsfonts}
\usepackage{filecontents}
\usepackage{bm}
\usepackage{amsthm}
\usepackage{algorithm}
\usepackage{multicol}
\usepackage{setspace}
\usepackage[noend]{algpseudocode}
\usepackage{subfigure}
\usepackage{tabularx}
\usepackage{textcomp}
\usepackage{multirow}
\newtheoremstyle{exampstyle}
{3pt} 
{2pt} 
{} 
{} 
{\bfseries} 
{.} 
{.5em} 
{} 
\theoremstyle{exampstyle}
\newtheorem{problem}{Problem}

\newtheorem{proposition}{Proposition}

\newtheorem{remark}{Remark}
\normalsize
\def\BibTeX{{\rm B\kern-.05em{\sc i\kern-.025em b}\kern-.08em T\kern-.1667em\lower.7ex\hbox{E}\kern-.125emX}}
\begin{document}

\title{A Deep Learning-Based Approach for Cell Outage Compensation in NOMA Networks}

\author{Elaheh Vaezpour, Layla Majzoobi, Mohammad Akbari, Saeedeh Parsaeefard, Halim Yanikomeroglu, \IEEEmembership{Fellow, IEEE}
\thanks{Elaheh Vaezpour is with Communication Department, ICT Research Institute, Tehran, Iran (e-mail: e.vaezpour@itrc.ac.ir).}
\thanks{Layla Majzoobi is Research Assistant at Communication Department, ICT Research Institute, Tehran, Iran (e-mail: l.majzoobi@itrc.ac.ir).}
\thanks{Mohammad Akbari is with Communication Department, ICT Research Institute, Tehran, Iran (e-mail:  m.akbari@itrc.ac.ir).}
\thanks{Saeedeh Parsaeefard is with the Department of Electrical and Computer Engineering, University of Toronto, Toronto, Canada (email: saeideh.fard@utoronto.ca).}
\thanks{Halim Yanikomeroglu is with the Department of Systems and Computer Engineering, Carleton University, Ottawa, Canada (email: halim@sce.carleton.ca)
}}

\IEEEtitleabstractindextext{\begin{abstract}Cell outage compensation enables a network to react to a catastrophic cell failure quickly and serve users in the outage zone uninterruptedly. Utilizing the promising benefits of non-orthogonal multiple access (NOMA) for improving the throughput of cell edge users, we propose a newly NOMA-based cell outage compensation scheme. In this scheme, the compensation is formulated as a mixed integer non-linear program (MINLP) where outage zone users are associated to neighboring cells and their power are allocated with the objective of maximizing spectral efficiency, subject to maintaining the quality of service for the rest of the users. Owing to the importance of immediate management of cell outage and handling the computational complexity, we develop a low-complexity suboptimal solution for this problem in which the user association scheme is determined by a newly heuristic algorithm, and power allocation is set by applying an innovative deep neural network (DNN). The complexity of our proposed method is in the order of polynomial basis, which is much less than the exponential complexity of finding an optimal solution. Simulation results demonstrate that the proposed method approaches the optimal solution. Moreover, the developed scheme greatly improves fairness and increases the number of served users.

\end{abstract}

}

\maketitle

\section{INTRODUCTION}

\IEEEPARstart{G}{rowing} demand for very high speed and low latency wireless services is making cellular networks denser and more complex. The planning and management of networks with a huge number of elements and parameters is one of the biggest challenges that cellular network operators face. Owing in part to this, inefficiencies in manual solutions has lead self-organizing networks (SONs) to become an essential part of the management of wireless cellular networks \cite{TS32500}. SON technology aims to reduce capital and operational expenditures by minimizing human intervention in a network by means of various functionalities, including self-configuration, self-optimization, and self-healing \cite{aliu2012survey}. Cell outage management is one of the most important use cases of self-healing. It applies to BSs that can no longer serve users in its zone, which leads to a coverage hole in the network \cite{TS32541}. In cell outage management, compensation refers to the parameter adjustment of the cells surrounding the outage zone, which is intended to minimize the impact of the outage on the network.

 Recently, non-orthogonal multiple access (NOMA) has been proposed in 5G networks to improve the throughput of cell-edge users \cite{islam2016power} and the spectral efficiency (SE) of networks by supporting more users than the number of available orthogonal resources \cite{saito2013non}.
 Due to NOMA's ability to boost performance for cell edge users, it can be a good candidate in cell outage compensation. In this scenario, users in the outage zone (failed users) are indeed the cell edge users of the surrounding cells taking part in the compensation process. To fully utilize the benefit of the proposed NOMA-based compensation scheme, we model the compensation process as a mixed integer non-linear program (MINLP) where there are two key problems: (1) optimally assigning failed users to neighboring cells; and (2) allocating power in order to maximize the sum of the failed users' SE while maintaining quality of service (QoS) for users in neighboring cells.
To the best of our knowledge, this is the first work that attempts to mathematically model the joint problem of failed user association and power allocation in a NOMA-based cell outage compensation scheme. This problem is computationally difficult to solve \cite{kannan1978computational}. We decompose the original problem into two sub-problems that we solve sequentially (i.e., user association and then power allocation), and we propose effective low complexity solutions for each sub-problem.
\subsection{Related Work}
Several proposals for cell outage compensation frameworks can be found in the literature. Here, we briefly review some of the most important existing works in this area.

The authors in \cite{amirijoo2011effectiveness} studied the effectiveness of different parameter adjustments, such as reference signal power, antenna tilt, and scheduling parameters in cell outage compensation. In \cite{de2017improving}, a guideline to improve compensation process is provided. In \cite{qin2018machine,lee2013cobra,lee2015low}, orthogonal frequency-division multiple access
(OFDMA)-based compensation methods were proposed where the total bandwidth of each cell was split into two parts: a part used to transmit data for normal users (normal bandwidth), and a part used to serve users in an outage zone (healing bandwidth). The proposed algorithm in \cite{lee2013cobra} utilized cooperative beamforming in the healing channel and formulated the compensation process as an optimization problem. The objective was to find an optimal subchannel assignment and power allocation scheme to maximize the weighted sum-rate of users in the network. In \cite{lee2015low}, a low-complexity resource allocation algorithm based on an optimization approach was proposed for cell outage compensation. The aim was to maximize sum-rate in the network while providing all users with a minimum rate requirement. In \cite{qin2018machine}, compensation was modeled as a utility maximization problem and a distributed algorithm was proposed to solve it. However, the authors did not consider the performance degradation for users in the neighboring cells due to the compensation nor did they consider bandwidth splitting. The use of Unmanned Aerial Vehicles (UAV) for compensation of cell outage is studied in \cite{selim2018short}. In \cite{onireti2015cell}, an actor critic (AC) based reinforcement learning (RL) cell outage compensation algorithm was proposed to tune power and antenna tilt of the neighboring cells. A deep RL-based approach for outage compensation was proposed in \cite{guo2019deep} where association of users in the outage zone to neighboring cells was done through a K-means clustering algorithm; and for compensation the antenna downtilt and user power was determined by a Deep Q-learning method. In \cite{chen2021cell} a cell outage compensation mechanism based on an improved particle swarm algorithm (IPSO) is proposed. In this work, IPSO is adopted based on adaptive weights to comprehensively adjust the downtilt, horizontal azimuth, and user power allocation of the neighboring base station to ensure the QoS requirement of the users. In \cite{jie2020comp}, a cell outage compensation mechanism using a network cooperation scheme in a heterogeneous network with densely deployed Femto Base Stations
(FBSs) is proposed. The proposed scheme utilizes the Coordinated Multi-Point (CoMP) transmission
and reception with joint processing technique where the main objective is to
guarantee service to UEs even in cases where their primary FBSs or their backhauls fail. In spite of all the above efforts, cell outage compensation in NOMA-based networks has been overlooked and needs to be studied. Therefore, we do so by considering the problem of joint power allocation and user association.

The problem of power allocation in a NOMA-based system can be solved by traditional optimization-based numerical approaches. In \cite{ali2016dynamic}, power allocation in uplink and downlink were cast as convex optimization problems, and closed-form solutions were proposed for them. An imperfect NOMA scheme suffering from receiver sensitivity and interference residue from non-ideal decoding was analyzed in \cite{celik2019distributed}, where the objective of resource allocation optimization problem captured a trade-off between maximum throughput and proportional fairness, and an iterative algorithm was developed to solve it. However, these optimization-oriented methods need to be performed repeatedly whenever network characteristics (such as a change in propagation scenarios) vary over time. Therefore, they incur considerable online computational cost and high overhead that impair their real-time development. Meta-heuristic algorithms such as genetic algorithms \cite{ma2017power} and particle swarm optimization \cite{xiao2018improved} have also been used to deal with the power allocation problem in NOMA systems. However, they have the same drawbacks as the optimization-based numerical approaches mentioned above.

Taking these limitations into account, learning-based approaches have recently been proposed. One popular such approach used to deal with the problem of unknown varying environments is RL \cite{sutton2018reinforcement}. In \cite{onireti2015cell} and \cite{guo2019deep}, RL-based methods were proposed for power allocation in an outage compensation scenario. Yet one drawback of RL-based approaches is that they try to find the optimal solution by interacting with the environment through trial and error and may require a long time to converge. This makes them less attractive in practical implementations of cell outage compensation scenarios where the network needs to react quickly to the catastrophic failure and provide users with service. Unsupervised learning algorithms do not leverage the data that can be achieved offline during the operation of a network. To utilize the data that can be obtained offline and achieve much lower online computational complexity, we propose a supervised deep neural network (DNN) for the cell outage compensation problem. DNNs have been shown to be a promising approach for approximating the optimal policy for different resource allocation problems in wireless networks, such as user association \cite{zappone2018user}, power control  \cite{liang2019towards,sun2018learning,yang2019deep}, subchannel assignment \cite{ahmed2019deep}, and beamforming \cite{alkhateeb2018deep}. Although several works have focused on DNN-based resource allocation in wireless networks, none have investigated its potential in dealing with the problem of cell outage compensation, which is addressed in this paper.
\subsection{Contributions and Paper Organization}

 The key contributions of this work can be summarized as follows:
\begin{itemize}
  \item A NOMA-based cell outage compensation scheme is proposed.
  \item The proposed scheme manages the cell outage with the least required changes in the users connections, which in turn imposes minimum signaling overhead on the network.
  \item A new formulation for joint failed user association and power allocation is presented as an optimization problem.
  \item For the NP-hard joint failed user association and power allocation problem, a newly and efficient two-step sequential approach is proposed. The computational complexity of this approach is in the order of polynomial basis. Our proposed approach follows the concept of learning to optimize where we use deep neural networks (DNN) to provide the best map for the power allocation step of our algorithm.
\end{itemize}

The rest of the paper is organized as follows: The system model and problem formulation are introduced in Section \ref{sec:SysModel} and \ref{sec:ProblemFormulation}, respectively. The proposed solution methodology is presented in Section \ref{sec:solution} and Section \ref{sec:cochannelInt}. In Section \ref{sec:MIMO}, we will explain how our proposed method can be extended to MIMO case. Section \ref{sec:simulation} demonstrates the evaluation of the proposed scheme, and conclusions are drawn in Section \ref{sec:conclusion}.
\section{System Model}\label{sec:SysModel}
Consider a two-tier cellular network architecture, where control and data planes are separated \cite{mohamed2015control} and \cite{xu2013functionality}. The control plane, constructed of macro cells, handles user connectivity and signaling procedures, including synchronization, broadcast, multicast, paging, and RRC functionalities. The data plane consists of small cells and provides the connectivity facilities with high data rates for end-users. In this context, the two layers operate in different sub-bands, which prevents co-channel interference between the data plane and control plane cells. In this paper, we focus on the management of the data plane, where a single input single output downlink NOMA-based network is considered. It is assumed that the neighboring cells transmit on different sub-bands, and therefore, the co-channel interference between them is negligible. Each base station (BS) operates in power-domain NOMA where users are divided into multiple clusters. Users in each cluster share the same time-frequency resources, which are orthogonal to the resources allocated to the other clusters. We assume that perfect channel-state information (CSI) is available at the transmitter.

Power domain NOMA allows assigning one transmission channel to multiple user in a cluster, and exploits the channel gain difference between them for the purpose of multiplexing \cite{maraqa2020survey}. It superposes users signal in the power domain at the transmitter. The superposition is performed such that each NOMA receiver can successfully decode the desired signal by applying a particular successive interference cancellation (SIC) technique at the corresponding receiver \cite{islam2016power}. A base station transmits the superposed signal to users with different powers allocated to each user.
For simplicity of explanation, let us consider a two-user downlink scenario. The base station transmit a signal $x_i$ to user $i$, where $i=1,2$ with power $p_i$. The transmitted signal is
\begin{equation*}
x = \sqrt{P_1}x_1 + \sqrt{P_2}x_2.
\end{equation*}
The received signal at user $i$ is
\begin{equation*}
y_i = h_i x + w_i,
\end{equation*}
where $h_i$ is the channel coefficient between the base station and user $i$. $w_i$ is the Gaussian noise, and its power density is $N_i$. Prior to decoding the desired signal, the receiver cancels the strong interference signals of the users with lower channel gains than the considered receiver. The signals for users with higher channel gains is considered as inter-user interferences.
\\
Hence, the corresponding SINR depends on the power allocated to a user and sum of the powers allocated to those having higher channel gains. Assuming that $h_1>h_2$, the SINR in the two-user case would be
\begin{equation*}
\sigma_1 = \frac{P_1 h_1}{N_1}, \quad \quad \sigma_2 = \frac{P_2 h_2}{N_2 + P_1 h_2}.
\end{equation*}
The power allocation mechanism in NOMA is performed in such way that higher transmission power is allocated to users with lower channel gain and vice versa. Therefore, the strongest interference that a user senses is due to the high powers allocated to users with weaker channel gains which is removed by using a proper SIC technique. To successfully decode the desired signal, a minimum difference is required between a user's signal and non-decoded inter-user interference.

We consider a cell outage scenario in which one of the BSs experiences outage during normal network operations and can no longer serve the users in its coverage area. Our objective is to compensate for the failure of this BS via serving its users by some or all of its neighboring cells. The neighboring cells participating in the compensation process are called compensating cells. We assume that users in the outage zone can receive the pilot of compensating cells \cite{lee2013cobra} meaning that they can estimate compensating cells channel gains and send the measurements to the OAM unit which manages outage compensation procedure. At the time of failure, the outage BS was serving a set $\mathcal{U}^f$ of $U^f$ users (failed users), and cell $n$ of compensating set $\mathcal{N}$ was serving a set $\mathcal{U}^n$ of $U^n$ users (connected users). The BS $n \in \mathcal{N}$ groups its $U^n$ users into a set $\mathcal{L}^n$ of  $L^n$ NOMA clusters and serves them using a set of orthogonal subcarriers. It is assumed that all active users in the network are provided with the same portion of available bandwidth. Figure \ref{fig:SystemModel} demonstrates the network model. In this figure, the outage BS and its users are depicted in red. As we can see, this BS has three  neighbors each of which serves a number of connected users in a NOMA-based system. We assume that the cell outage compensation process is managed by a central unit called "network
operation and management" (OAM).

  When a BS experiences outage, the network enters the compensation state. In this state, the objective is to serve users in the coverage zone of the outage BS by its compensating cells while minimizing the impact on the normal operation of connected users. To achieve objective, we aim to efficiently assign failed users to clusters of compensating cells and allocate power to all users while serving connected users without interruption and preserving their minimum QoS requirements. The compensation process can be cast as an optimization problem, which is described in the next section.

\begin{remark}\label{rem:objective}
In order to minimize the influence of the compensation process on the  network, we assume that connected user association schemes are the same before and after the cell outage.
\end{remark}

\begin{figure}[!t]
	\centering
	\centerline{\includegraphics[width=0.90\linewidth]{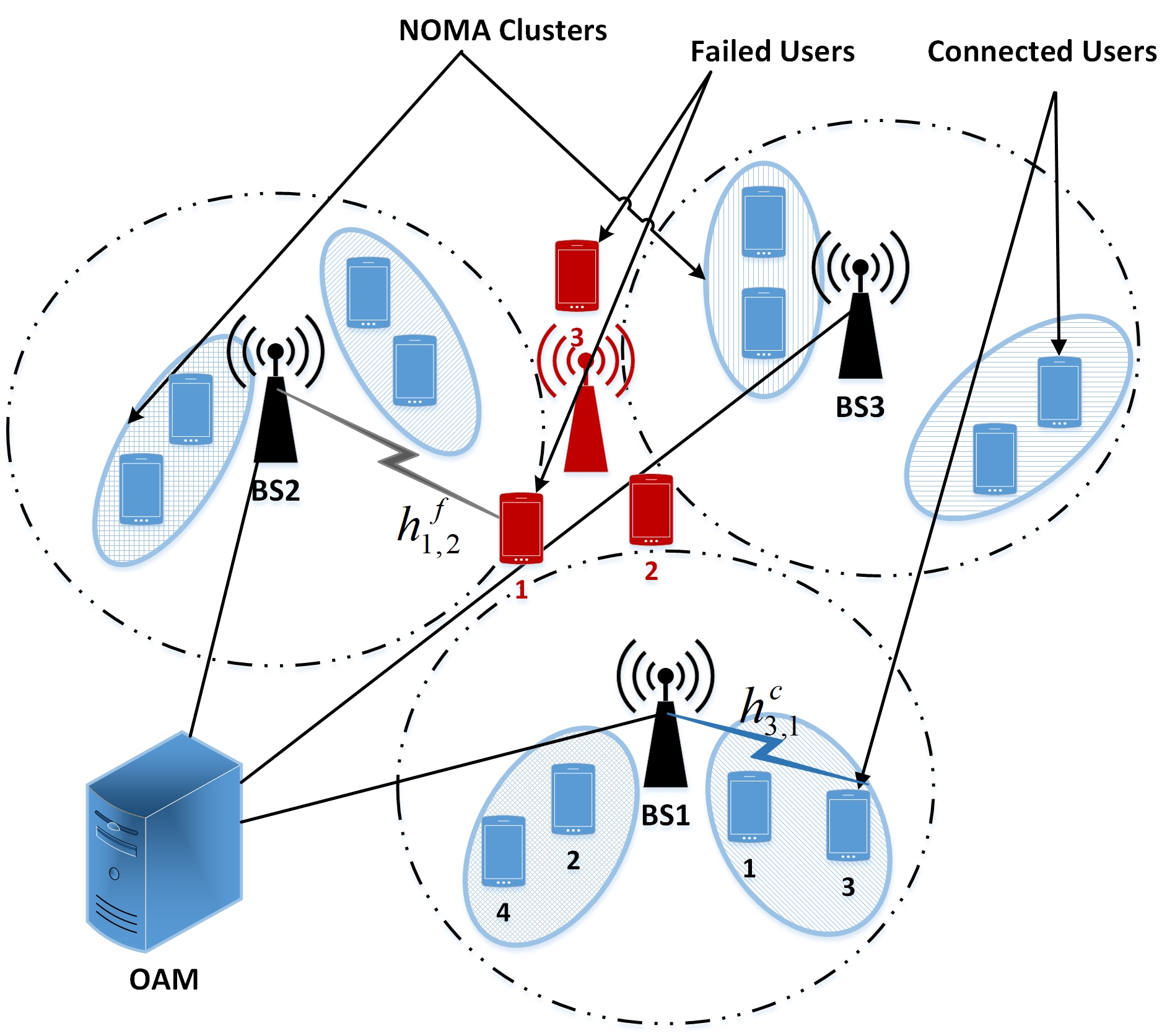}}
	\caption{ Illustration of the network model where the outage BS and its users are depicted in red.}
	\label{fig:SystemModel}
\end{figure}

\section{Problem Formulation}\label{sec:ProblemFormulation}

 In this section, we formulate the joint failed user association and power allocation problem with the objective of maximizing the sum of SE of the failed users subject to maintaining the QoS of other connected users in the system.

 Let us define binary parameter $\beta_{u,\ell,n}$ as the connected user clustering parameter for all $n \in \mathcal{N}$:
\begin{equation}\label{eq:betaDef}
\begin{split}
\beta_{u,\ell,n} &=
 \begin{cases}
    1,  &  \text{if } u \in \mathcal{U}^n \text{ is assigned to cluster } \ell \in \mathcal{L}^n, \\
   0,  &  \text{otherwise}.
  \end{cases}\\
  \end{split}
\end{equation}
 As mentioned in Remark \ref{rem:objective}, $\beta_{u,\ell,n}$ parameters do not change after cell outage compensation. Without loss of generality, we assume that connected users in each compensating cell are sorted according to the descending order of their channel gain as $h_{1,n}^{c} \geq h_{2,n}^{c} \geq \dots \geq h^{c}_{U^n,n}$, where superscript $c$ denotes the connected users of compensating cell $n$. Similarly, let sort the failed users according to their channel gains in descending order as $h_{1,n}^{f} \geq h_{2,n}^{f} \geq \dots \geq h^{f}_{U^f,n}$, where $h_{i,n}^{f}$ represents the $i^{\text{th}}$ largest channel gain among failed users with respect to the $n^{\text{th}}$ compensating cell. Since users in the outage zone are typically further from the compensating BS than their connected users, they experience weaker channel gain than their connected users, and we have
\begin{equation}\label{eq:FailedPrCH}
  h_{1,n}^{c} \geq \dots \geq h^{c}_{U^n,n} \geq h_{1,n}^{f} \geq \dots \geq h^{f}_{U^f,n}, \hspace{2mm} \forall n \in \mathcal{N}.
\end{equation}
Taking into account that the SIC decoding algorithm resolves interference from users with weaker channel gain \cite{ali2016dynamic}, the minimum SE requirement of connected users can be written as follows for all $u \in \mathcal{U}^n$ and $n \in \mathcal{N}$:
  \vspace{+0.3cm}
 \begin{equation*}\label{eq:minRate}
 \begin{split}
   \text{\textbf{C1}}:\sum_{\ell \in \mathcal{L}^n} \beta_{u,\ell,n} \log_2\bigg(1 + &\frac{p^{c}_{u,\ell,n}h^{c}_{u,n}}{h^{c}_{u,n}\sum\limits_{k =1}^{u-1} \beta_{k,\ell,n}p^{c}_{k,\ell,n} + \sigma^2}\bigg) \geq \mathfrak{s}^{\text{min}}_u, \\
   \end{split}
 \end{equation*}
 \normalsize
 where $p^{c}_{u,\ell,n}$ is the power allocated to connected user $u \in \mathcal{U}^n$ if it belongs to cluster $\ell \in \mathcal{L}^n$, $\sigma^2$ is the noise power, and $ \mathfrak{s}^{\text{min}}_u$ is the minimum SE requirement of user $u$.

 Let $\alpha_{u,\ell,n}$ be a binary variable indicating if failed user $u \in  \mathcal{U}^f$ is associated with cluster $\ell \in \mathcal{L}^n$, for all $n \in \mathcal{N}$:
   \vspace{+0.3cm}
\begin{equation}\label{eq:alphaDef}
\begin{split}
\alpha_{u,\ell,n} &=
  \begin{cases}
    1,       &  \text{if } u \in \mathcal{U}^f \text{ is assigned to cluster } \ell \in \mathcal{L}^n,\\
   0,  &  \text{otherwise}.\\
  \end{cases}\\
  \end{split}
\end{equation}\normalsize
 In order to reduce the computational complexity and SIC decoding delay at the receivers \cite{saito2013non} and \cite{wei2017optimal}, each cluster in the  compensating cells  is allowed to serve at most one failed user from the outage zone. As a consequence, we have $U^f \leq  \sum_{n \in \mathcal{N}}^{} L ^n$, and
   \vspace{+0.3cm}
 \begin{equation*}\label{eq:alphaConsCluster}
   \text{\textbf{C2}}:\sum_{u \in \mathcal{U}^f} \alpha_{u,\ell,n} \leq 1, \hspace{2mm} \forall n \in \mathcal{N}, \hspace{2mm}\forall \ell \in \mathcal{L}^n .
 \end{equation*}\normalsize
It should be noted that applying techniques such as load balancing results in uniform distribution of users among neighboring cells. In this way, we can assume that the number of users in the failed cell and in each of compensating cells are almost the same, i.e., $U^f \approx U^n, \forall n \in \mathcal{N}$. Considering NOMA clusters of $q$ users, we have $ L^n = U^n/q \approx U^f/q$. Assuming $N \geq q$, we have $NL^n \geq U^f$, which means that the number of failed users is not greater than the total number of clusters in the compensating cells. Due to the practical size of NOMA clusters and the number of available compensating cells (which can be considered equal to the number of neighbors of the failed cell), $ N \geq q$ is a practical assumption, and hence Constraint C2 is reasonable. In addition, in cases where $N <q$, we need an admission control policy. However, this is beyond the scope of this paper and will be left for future work.

We also assume that each failed user is allowed to be associated with exactly one NOMA cluster in the network, which means that
  \vspace{+0.3cm}
\begin{equation*}\label{eq:alphaCons}
  \text{\textbf{C3}}:\sum_{n \in \mathcal{N}}^{}\sum_{\ell \in \mathcal{L}^n} \alpha_{u,\ell,n} = 1, \hspace{2mm} \forall u \in \mathcal{U}^f.
\end{equation*}\normalsize
In order to successfully apply the SIC algorithm to both connected and failed users, the following conditions need to be satisfied \cite{ali2016dynamic}:
  \vspace{+0.3cm}
\begin{equation*}\label{eq:SICConF}
\begin{split}
 &\text{\textbf{C4}}:\big( \alpha_{u,\ell,n} p^{f}_{u,\ell,n} - \sum_{k \in \mathcal{U}^n} \beta_{k,\ell,n} p^{c}_{k,\ell,n} \big) H^{}_{u-1,\ell,n} \geq\\
  &\hspace{8mm}\mathfrak{p}^{\text{tol}} - (1-\alpha_{u,\ell,n})B,\hspace{1mm}\forall n \in \mathcal{N},\hspace{1mm}\forall \ell \in \mathcal{L}^n,  \hspace{1mm}\forall u \in \mathcal{U}^f\\
 \end{split}
 \end{equation*}
 \begin{equation*}\label{eq:SICConC}
 \begin{split}
 &\text{\textbf{C5}}:\big( \beta_{u,\ell,n} p^{c}_{u,\ell,n} - \sum\limits_{k =1}^{u-1} \beta_{k,\ell,n} p^{c}_{k,\ell,n} \big) H_{u-1,\ell,n}^{} \geq \\
 &\hspace{8mm}\mathfrak{p}^{\text{tol}} - (1-\beta_{u,\ell,n})B,\hspace{1mm}\forall n \in \mathcal{N},\hspace{1mm}\forall \ell \in \mathcal{L}^n, 2\leq u \leq U^n,\\
 \end{split}
\end{equation*}\normalsize
where $p^{f}_{u,\ell,n}$ is the power allocated to failed user $u \in \mathcal{U}^f$ if it belongs to cluster $\ell \in \mathcal{L}^n$, and $\mathfrak{p}^{\text{tol}}$ is the minimum power difference required to distinguish between the signal to be decoded and the remaining non-decoded message signals. $B$ is a sufficiently large number that is used to ensure that these constraints are non-binding  when  $\alpha_{u,\ell,n} =0$ and/or  $\beta_{u,\ell,n} = 0$.   For all $n \in \mathcal{N}$,  $H^{}_{u-1,\ell,n}$ is the minimum channel gain of the users in cluster $\ell \in \mathcal{L}^n$, which is greater than the channel gain of user $u$, and can be determined as\small
  \vspace{+0.3cm}
\begin{equation}\label{eq:Hu-1}
H^{}_{u-1,\ell,n} =
\begin{cases}
    \underset{1\leq k\leq u-1}{\operatorname{min}}\{h_{k,n}^{c}\beta_{k,\ell,n} | \beta_{k,\ell,n} = 1\}, &\text{if } u\in\mathcal{U}^n  , u \neq 1, \\
   \underset{1\leq k\leq U^n}{\operatorname{min}}\{h_{k,n}^{c}\beta_{k,\ell,n} | \beta_{k,\ell,n} = 1\}, &\text{if } u \in \mathcal{U}^f.\\
  \end{cases}
  \end{equation}
\normalsize
 From a practical point of view, we need to put a limit of $\mathfrak{p}^{\text{max}}$ on the available power at each BS $n \in \mathcal{N}$  as follows:
   \vspace{+0.3cm}
  \begin{equation*}\label{eq:Pmax}
  \begin{split}
    \text{\textbf{C6}}: \sum_{\ell \in \mathcal{L}^n}\sum_{u \in \mathcal{U}^f}  \alpha_{u,\ell,n}p_{u,\ell,n}^{f}
+\sum_{\ell \in \mathcal{L}^n} \sum_{u \in \mathcal{U}^n} \beta_{u,\ell,n} p^{c}_{u,\ell,n} &\leq \mathfrak{p}^{\text{max}}.\\
 \end{split}
  \end{equation*}
\normalsize
In the compensation process, the objective function is maximizing sum of SE of the failed users. From \eqref{eq:FailedPrCH}, since channel gain of failed users is weaker than the channel gain of connected users, failed users cannot cancel interference from connected users. Hence, the achievable SE of a failed user $u \in \mathcal{U}^f$ is
  \vspace{+0.3cm}
\begin{equation}\label{eq:rateFailedUser}
\sum_{n \in \mathcal{N}} \sum_{\ell \in \mathcal{L}^n} \alpha_{u,\ell,n} \log_2\bigg(1 + \frac{p^{f}_{u,\ell,n}h^{f}_{u,n}}{h^{f}_{u,n} \sum\limits_{k=1}^{U^n} \beta_{k,\ell,n} p^{c}_{k,\ell,n}  + \sigma^2}\bigg).
\end{equation}
\normalsize
Considering conditions \textbf{C1\hspace{1mm}}-\textbf{\hspace{1mm}C6}, the joint user association and power allocation optimization problem in the outage scenario with the objective of maximizing sum of SE of the failed users can be formulated as Problem \ref{pr:joint}.
\begin{problem}\label{pr:joint}
Given the channel gains of failed and connected users in the network and clustering parameter $\beta_{u,\ell,n}$ for all $n \in \mathcal{N}$, $\ell \in \mathcal{L}^n$ and $u \in \mathcal{U}^n$, the joint user association and power allocation optimization problem in the outage compensation scenario is formulated as follows:\small
\begin{equation*}\label{eq:JointAfter}
\begin{split}
\displaystyle &\max_{\mathcal{A} , \mathcal{P}^{f}, \mathcal{P}^{c}}  \sum_{u \in \mathcal{U}^f} \sum_{n \in \mathcal{N}} \sum_{\ell \in \mathcal{L}^n} \alpha_{u,\ell,n} S_{u,\ell,n}, \\
&\text{subject to:}\\
&\text{\textbf{C1\hspace{1mm}}}- \text{\textbf{C6\hspace{1mm}}}, \\
&\text{\textbf{C7}}: p^{f}_{u,\ell,n} \geq 0, \hspace{2mm} \forall n \in \mathcal{N},\hspace{2mm} \forall \ell \in \mathcal{L}^n, \hspace{2mm} \forall u \in \mathcal{U}^f,\\
&\text{\textbf{C8}}: p^{c}_{u,\ell,n} \geq 0, \hspace{2mm} \forall n \in \mathcal{N}, \hspace{2mm} \forall \ell \in \mathcal{L}^n, \hspace{2mm} \forall u \in \mathcal{U}^n,\\
&\text{\textbf{C9}}: \alpha_{u,\ell,n} \in \{0,1\}, \hspace{2mm} \forall n \in \mathcal{N}, \hspace{2mm} \forall \ell \in \mathcal{L}^n, \hspace{2mm} \forall u \in \mathcal{U}^f,
\end{split}
\end{equation*}\normalsize
where

\begin{equation*}
  S_{u,\ell,n} = \log_2\bigg(1 + \frac{p^{f}_{u,\ell,n}h^{f}_{u,n}}{h^{f}_{u,n} \sum\limits_{k=1}^{U^n} \beta_{u,\ell,n} p^{c}_{k,\ell,n}  + \sigma^2}\bigg),
\end{equation*}
and $\mathcal{A}$,$\mathcal{P}^{f}$ and $\mathcal{P}^{c}$ are sets whose elements are $\alpha_{u,\ell,n}$, $ p^{f}_{u,\ell,n}$ and $p^{c}_{u,\ell,n}$, respectively.
Constraints \textbf{C7} and \textbf{C8} guarantee the nonnegativity of power allocated to all users, and constraint \textbf{C9} defines variable $\alpha_{u,\ell,n}$ as a binary one.
\end{problem}
As we can see, Problem \ref{pr:joint} is a MINLP, which is an NP-hard combinatorial problem \cite{kannan1978computational}.
Accordingly, finding the optimal solution for user association and power allocation requires an exhaustive search that incurs exponential computational complexity. In the following, we propose an efficient low-complexity solution for this problem.

\section{The Proposed Solution}\label{sec:solution}
In this section, we propose a two-step low-complexity sequential solution for Problem \ref{pr:joint}. In the first step, we heuristically find a suboptimal user association scheme, and in the second step, the optimal power allocation is found for the user association scheme from the first step.
\subsection{User Association}\label{sec:UA}
Here, we present our heuristic suboptimal failed user association algorithm. In our proposed algorithm, the goal is to assign failed users to the clusters that can provide them with higher SE. Hence, we need to estimate the amount of power each cluster can allocate to failed users while still fulfilling the QoS requirements of connected users. To do so, we need to know the clustering and power allocation policy of connected users in the compensating cell before the cell outage event. Therefore, we first describe this policy and then present our failed user association algorithm.

\subsubsection{Power allocation and user association before the outage}\label{sec:PABefore}
 We assume that prior to the outage users at BS $n$ are clustered on the basis of the downlink clustering algorithm proposed in \cite{ali2016dynamic}. The aim of the algorithm is to put users with high channel gain in different clusters and pair them with those whose channel gain is low. In addition, we assume that at BS $n \in \mathcal{N}$, optimal power allocation is done to maximize the sum of SE of its users while considering their minimum SE requirement as follows:

\small
\begin{equation}\label{eq:PABefor}
\begin{split}
 \displaystyle &\max_{\mathcal{P}^{'c}_n} \sum_{u \in \mathcal{U}^n} \sum_{\ell \in \mathcal{L}^n} \beta_{u,\ell,n} S_{u,\ell,n}^{'c},\\
&\text{subject to:}\\
&\text{\textbf{C1\hspace{1mm}:\hspace{1mm}}} \sum_{\ell \in \mathcal{L}^n} \beta_{u,\ell,n} S_{u,\ell,n}^{'c} \geq \mathfrak{s}^{\text{min}}_u, \forall u \in \mathcal{U}^n, \\
& \text{\textbf{C2\hspace{1mm}:\hspace{1mm}}} \big( \beta_{u,\ell,n} p^{'c}_{u,\ell,n} - \sum\limits_{k =1}^{u-1} \beta_{k,\ell,n} p^{'c}_{k,\ell,n} \big) H_{u-1,\ell,n}^{c}
 \geq\\
 &\hspace{20mm} \mathfrak{p}^{\text{tol}} - (1-\beta_{u,\ell,n})M, \hspace{2mm} 2 \leq u \leq U^n,\hspace{2mm}\forall \ell \in \mathcal{L}^n, \\
& \text{\textbf{C3\hspace{1mm}:\hspace{1mm}}} \sum_{\ell \in \mathcal{L}^n} \sum_{u \in \mathcal{U}^n} \beta_{u,\ell,n} p^{'c}_{u,\ell,n} \leq \mathfrak{p}^{\text{max}},\\
& \text{\textbf{C4\hspace{1mm}:\hspace{1mm}}} p^{'c}_{u,\ell,n} \geq 0, \hspace{2mm} \forall u \in \mathcal{U}^n, \hspace{2mm} \forall \ell \in \mathcal{L}^n,\\
\end{split}
\end{equation}\normalsize
where
 $p^{'c}_{u,\ell,n}$ is the power that BS $n$ allocates to user $u$ which belongs to cluster $\ell \in \mathcal{L}^n$ at the time of failure, $\mathcal{P}^{'c}_n$ is a set whose elements are $p^{'c}_{u,\ell,n}$, and
  $$S_{u,\ell,n}^{'c} = \log_2\bigg(1 + \frac{p^{'c}_{u,\ell,n}h^{c}_{u,n}}{  h^{c}_{u,n}\sum\limits_{k =1}^{u-1} \beta_{k,\ell,n}p^{'c}_{k,\ell,n}  + \sigma^2}\bigg).$$
  Constraint \textbf{C1} guarantees the minimum SE requirement for the users. Constraint \textbf{C2} denotes the necessary power constraint for efficient SIC decoding in a NOMA cluster. Constraint \textbf{C3} is total power constraint of each BS, and constraint \textbf{C4} ensures nonnegativity of power allocated to each user.

\subsubsection{Failed user association algorithm}\label{sec:failedAssociation}
According to Problem \ref{pr:joint}, after the outage, all clusters allocate all of their power to the failed users except what is required to satisfy the SIC decoding and the minimum SE requirement of connected users. Knowing the amount of power required by connected users of each cluster, we can determine how large the SE of a failed user would be if it connected to that cluster.We can then assign each failed user to the cluster that offers the highest SE.

 As Problem \ref{pr:joint} indicates, the minimum power that needs to be allocated to the user with the highest channel gain in cluster $\ell \in \mathcal{L}^n$, which is denoted as $u_{m^\ell}$ in the sequel, can be determined on the basis of its QoS requirement as follows:
 \begin{equation}\label{eq:minPowerMaxUser}
   p_{u_{m^\ell},\ell,n}^{c} = \frac{\sigma^2}{h_{u_{m^\ell},n}^{c}}\left( 2^{\mathfrak{s}^{\text{min}}_{u_{m^\ell}}} - 1\right), \hspace{2mm} \forall n \in \mathcal{N}, \hspace{2mm}\forall \ell \in \mathcal{L}^n.
 \end{equation}
The amount of power that needs to be allocated to other connected users in each cluster can be approximated according to the following Proposition. In this proposition, following the literature in this context, we apply asymptotic analysis in inference-limited case.

\begin{proposition}\label{th:Power}
  Consider Problem \ref{pr:joint} and Subsection \ref{sec:PABefore}. We assume that the noise power $\sigma^2$  is negligible compared to the interference term, and $\mathfrak{p}^{\text{tol}} << 1$. The amount of power that needs to be allocated to connected users after the outage event, i.e., solution of Problem \ref{pr:joint}, can be approximated as
    \vspace{+0.3cm}
  \begin{equation}\label{eq:PApprox}
  p^{c}_{u,\ell,n} \approx \delta_{\ell,n}p^{'c}_{u,\ell,n}, \hspace{2mm}\forall n \in \mathcal{N}, \hspace{2mm}\forall \ell \in \mathcal{L}^n, \hspace{2mm}\forall u \in \mathcal{U}^n,
  \end{equation}
  where $p^{'c}_{u,\ell,n}$ is the solution of the problem in \eqref{eq:PABefor}, and
    \vspace{+0.3cm}
  \begin{equation}\label{eq:highestUserPower}
   \delta_{\ell,n} = \frac{p_{u_{m^\ell},\ell,n}^{c}}{p_{u_{m^\ell} , \ell,n}^{'c}}, \hspace{2mm} \forall n \in \mathcal{N}, \hspace{2mm}\forall \ell \in \mathcal{L}^n,
   \end{equation}
   where $p_{u_{m^\ell},\ell,n}^{c}$ can be calculated from \eqref{eq:minPowerMaxUser}.
\end{proposition}
 \begin{proof}
   Due to the assumption of negligibility of noise variance $\sigma^2$ compared to the interference term, prior to the outage the SE of connected user $u$ can be approximated as
   \vspace{+0.3cm}
 \begin{equation}\label{eq:NOMARate}
 \begin{split}
   S_{u,\ell,n}^{'c} &\approx \log_2\bigg(1 + \frac{p^{'c}_{u,\ell,n}h^{c}_{u,n}}{  h^{c}_{u,n}\sum\limits_{k =1}^{u-1} \beta_{k,\ell,n}p^{'c}_{k,\ell,n}}\bigg),\\
    &\hspace{8mm} \forall n \in \mathcal{N}, \hspace{2mm} \forall \ell \in \mathcal{L}^n, \hspace{2mm} u \in \mathcal{U}^n, \hspace{2mm}u \neq u_{m^\ell} .
   \end{split}
 \end{equation}
 In the compensation step, let  $p^{c}_{u,\ell,n} = \delta_{\ell,n}p^{'c}_{u,\ell,n}$. Then, the SE of user $u$ can be approximated as\small
 \vspace{+0.3cm}
  \begin{equation}\label{eq:SEAfter}
  \begin{split}
    S_{u,\ell,n}^{c} &\approx \log_2\bigg(1 + \frac{ \delta_{\ell,n}p^{'c}_{u,\ell,n}h^{c}_{u,n}}{  h^{c}_{u,n}\sum\limits_{k =1}^{u-1} \beta_{k,\ell,n} \delta_{\ell,n}p^{'c}_{k,\ell,n}}\bigg)\\
    &\stackrel{(a)}{=} S_{u,\ell,n}^{'c}, \hspace{2mm} \forall n \in \mathcal{N}, \hspace{2mm} \forall \ell \in \mathcal{L}^n, \hspace{2mm} u \in \mathcal{U}^n, \hspace{2mm} u \neq u_{m^\ell},
    \end{split}
  \end{equation}\normalsize
  where $(a)$ follows from \eqref{eq:NOMARate}. Note that according to \eqref{eq:FailedPrCH}, after the outage, connected users can resolve interference from the failed users, and we do not need to consider the interference from failed users in \eqref{eq:SEAfter}.

   According to \eqref{eq:SEAfter}, reducing  the power of all users in cluster $\ell \in \mathcal{L}^n$  by factor $\delta_{\ell,n}$ in the compensation step does not change their SE. The assumption of $\mathfrak{p}^{\text{tol}} << 1$ implies that this power allocation policy also meets the SIC decoding power constraints in \textbf{C5}. Consequently, \eqref{eq:PApprox} is an appropriate approximation for the first step of of the solution to Problem \ref{pr:joint}, which completes the proof.
 \end{proof}
 Now from \eqref{eq:PApprox},  we can propose an efficient scheme for the association of failed users to the clusters of compensating cells according to the following steps:
 \begin{itemize}
   \item  \textit{Determine the amount of power each cluster can allocate to a failed user}: According to Proposition \ref{th:Power}, the amount of power that cluster $\ell \in \mathcal{L}^n$ can allocate to a failed user can be calculated as
       \vspace{+0.3cm}
 \begin{equation}\label{eq:culsterPowerBudget}
 \Delta_{\ell,n} = (1 - \delta_{\ell,n})\sum_{u \in U^n}\beta_{u,\ell,n} p^{'c}_{u,\ell,n}, \hspace{2mm} \forall n \in \mathcal{N}, \hspace{2mm} \forall \ell \in \mathcal{L}^n.
 \end{equation}
   \item \textit{Determine the SE of each failed user if it is associated with cluster $\ell \in \mathcal{L}^n$}: According to \eqref{eq:PApprox} and \eqref{eq:culsterPowerBudget}, the SE of user $u \in \mathcal{U}^f$ if it connects to cluster $\ell \in \mathcal{L}^n$ is\small
       \vspace{+0.3cm}
       \begin{equation}\label{eq:failedUserRate}
       \begin{split}
   S^{f}_{u,\ell,n} &=
    \log_2\bigg(1 + \frac{\Delta_{\ell,n}h^{f}_{u,n}}{  h^{f}_{u,n}\sum\limits_{k =1}^{U^n} \beta_{u,\ell,n} \delta_{\ell,n}p^{'c}_{k,\ell,n}  + \sigma^2}\bigg),\\
     &\hspace{10mm} \forall n \in \mathcal{N}, \hspace{2mm} \forall \ell \in \mathcal{L}^n, \hspace{2mm} \forall u\in U^f.
     \end{split}
 \end{equation}\normalsize

   \item \textit{Determine the efficient failed user association scheme}: To maximize the sum of SE of the failed users, the user association scheme can be determined according to $S^{f}_{u,\ell,n}$ values in an iterative manner. At iteration $i$, triple $(u^*_i,\ell^*_i,n^*_i)$ is determined as follows:
\begin{equation}\label{eq:trp}
  (u^*_i,\ell^*_i,n^*_i) =  \underset{n,\ell,u}{\text{argmax}}\hspace{1mm}\{S^{f}_{u,\ell,n}\}.
\end{equation}
User $u^*_i \in \mathcal{U}^f$ will be associated with cluster $\ell^*_i$ of BS $n^*_i$. According to \textbf{C2} and \textbf{C3} in Problem \ref{pr:joint}, since each cluster can  at most serve one failed user, and also each user can be served just by one cluster, cluster $\ell^*_i $ and user $u^*_i$ should be removed from the candidate list in the next iterations. To do so, we can update the set $\{S^{f}_{u,\ell,n}\}$ at iteration $i$ as follows:
\begin{equation}\label{eq:SEUpdate1}
S^{f}_{u^*_i,\ell,n} = 0, \hspace{2mm} \forall n \in \mathcal{N}, \hspace{2mm}\forall \ell \in \mathcal{L}^n,
\end{equation}
\begin{equation}\label{eq:SEUpdate2}
S^{f}_{u,\ell^*_i,n^*_i} = 0, \hspace{2mm} \forall u \in \mathcal{U}^f.
\end{equation}
Due to the fact that there are $U^f$ failed users in the network, user association scheme can be determined in $U^f$ iterations.
 \end{itemize}

 This failed user association algorithm is outlined in Algorithm \ref{alg:userAssociation}.
 \begin{algorithm}[t]
	\caption{: Heuristic user association algorithm}\label{alg:userAssociation}
	Initialization: $\alpha_{u,\ell,n} = 0, \hspace{2mm} \forall u \in \mathcal{U}^f, n \in \mathcal{N}, \ell \in \mathcal{L}^n$.
	\begin{algorithmic}
		\ForAll{$n \in \mathcal{N}, \ell \in \mathcal{L}^n, u \in \mathcal{U}^f$ }
		\State Calculate $S^{f}_{u,\ell,n}$ using \eqref{eq:failedUserRate};
		\EndFor
        \ForAll{$i = 1, 2, \dots, U^f$}
        \State Determine $(u^*_i,\ell^*_i,n^*_i)$ according to \eqref{eq:trp};
        \State $\alpha_{u^*_i,\ell^*_i,n^*_i} = 1$;
        \State Update $S^{f}_{u,\ell,n}$ using \eqref{eq:SEUpdate1}-\eqref{eq:SEUpdate2};
        \EndFor
	\end{algorithmic}
\end{algorithm}

\subsection{Power Allocation}\label{sec:PA}
As discussed above, BSs operate on different frequencies and do not interfere with each other. Thus, given the user association scheme, i.e., $\mathcal{A}$,
the power allocation problem can be solved separately for the BSs that serve at least one failed user.
BSs that do not serve any failed users continue to operate as before.
Therefore, the power allocation problem for each BS $n$ can be written as Problem \ref{pr:power}.
\begin{problem}\label{pr:power}
Given the user association scheme, i.e., $\mathcal{A}$, and clustering parameter $\beta_{u,\ell,n}$ for all $\ell \in \mathcal{L}^n$ and $u \in \mathcal{U}^n$, the power allocation problem for each BS $n$ is formulated as follows:
\small
\begin{equation*}\label{eq:power_control}
  \begin{split}
  &\max_{\mathcal{P}^{f}, \mathcal{P}^{c}}  \sum_{u \in \mathcal{U}^f} \sum_{\ell \in \mathcal{L}^n} \alpha_{u,\ell,n} \log_2\bigg(1 + \frac{p^{f}_{u,\ell,n}h^{f}_{u,n}}{h^{f}_{u,n} \sum\limits_{k=1}^{U^n} \beta_{k,\ell,n} p^{c}_{k,\ell,n}  + \sigma^2}\bigg), \\
&\text{subject to:}\\
&\text{\textbf{C1}}: \sum_{\ell \in \mathcal{L}^n} \beta_{u,\ell,n} \log_2\bigg(1 + \frac{p^{c}_{u,\ell,n}h^{c}_{u,n}}{h^{c}_{u,n}\sum\limits_{k =1}^{u-1} \beta_{u,\ell,n}p^{c}_{k,\ell,n}  + \sigma^2}\bigg) \geq \mathfrak{s}^{\text{min}}_u,\\
 &\hspace{70mm} \forall u \in \mathcal{U}^n, \\
&\text{\textbf{C2}}:\big( \alpha_{u,\ell,n}  p^{f}_{u,\ell,n} - \sum_{k \in \mathcal{U}^n} \beta_{k,\ell,n} p^{c}_{k,\ell,n} \big) H^{}_{u-1,n} \geq \\
&\hspace{24mm}\mathfrak{p}^{\text{tol}} - (1-\alpha_{u,\ell,n})B,
\hspace{2mm} \forall u \in \mathcal{U}^f,  \hspace{2mm} \forall \ell \in \mathcal{L}^n,\\
&\text{\textbf{C3}}: \big( \beta_{u,\ell,n} p^{c}_{u,\ell,n} - \sum\limits_{k =1}^{u-1} \beta_{k,\ell,n} p^{c}_{k,\ell,n} \big) H_{u-1,\ell,n}^{} \geq\\
&\hspace{24mm} \mathfrak{p}^{\text{tol}} - (1-\beta_{u,\ell,n})B,
 \hspace{2mm} \forall u \in \mathcal{U}^n,  \hspace{2mm} \forall \ell \in \mathcal{L}^n, \\
&\text{\textbf{C4}}: \sum_{\ell \in \mathcal{L}^n}\sum_{u \in \mathcal{U}^f}  \alpha_{u,\ell,n}p_{u,\ell,n}^{f}
+\sum_{\ell \in \mathcal{L}^n} \sum_{u \in \mathcal{U}^n} \beta_{u,\ell,n} p^{c}_{u,\ell,n} \leq \mathfrak{p}^{\text{max}},  \\
&\text{\textbf{C5}}: p^{f}_{u,\ell,n} \geq 0, \hspace{2mm} \forall u \in \mathcal{U}^f,\\
&\text{\textbf{C6}}: p^{c}_{u,\ell,n} \geq 0, \hspace{2mm} \forall u \in \mathcal{U}^n.
  \end{split}
\end{equation*}\normalsize
\end{problem}
Constraints \textbf{C2}, \textbf{C3} , and \textbf{C4} are affine. The SE function of variables $p^{f}_{u,\ell,n}$ and $p^{c}_{k,\ell,n}$ in the objective function and constraint \textbf{C1} are concave functions. Since the proof procedure of the concavity of this function is similar to that in \cite[Lemma 1]{ali2018downlink}, it is omitted for the sake of brevity. Therefore, Problem \ref{pr:power} is a convex optimization problem. Mathematical-based methods inevitably rely
on network parameters and need to be performed for every system realization. This makes them a computational burden and limits their implementation in real time. Traditional iterative methods, such as the waterfilling algorithm, also suffer from computational complexity in practical implementation \cite{sun2018learning}. We propose exploiting DNN to
approximate this problem and utilize it to predict the power allocation scheme with much lower online complexity than optimization-oriented methods. We adopt a
feedforward fully connected DNN, which can exhibit a universal function approximation property \cite{hornik1989multilayer, leshno1993multilayer}.

A feedforward fully
connected DNN is composed of an input layer, multiple hidden layers, and one output layer of neurons. The general structure of this DNN is depicted in Figure \ref{fig:DNN}. An input vector of $D^{0}$ dimension is fed to the network through $D^{0}$ number of neurons in the input layer. Then, it passes the hidden layers, where each hidden layer $m$ has $D^{m}$ neurons. Finally, the output vector with dimension $D^{M}$ is achieved from the output layer $M$ with $D^{M}$ neurons. Each neuron $i$ in layer $k$ takes inputs from the previous layer $k-1$ and returns an output $z_{i,k}$ which is derived by the following equation:
\begin{equation}
z_{i,k} = \emph{f}_{i,k}(\sum\limits_{j=1}^{D^{k-1}}w_{j,k-1}z_{j,k-1}+b_{i,k}),
\end{equation}
where $\emph{f}_{i,k}$ is the activation function of neuron $i$ in layer $k$, $w_{j,k-1}$ denotes the weight of the output of neuron $j$ in the previous layer $k-1$, and $b_{i,k}$ denotes the bias term. The input to the activation function is the weighted sum of all the outputs from the previous layer plus a bias term.
 \begin{figure}[!t]
	\centering
	\centerline{\includegraphics[width=0.8\linewidth]{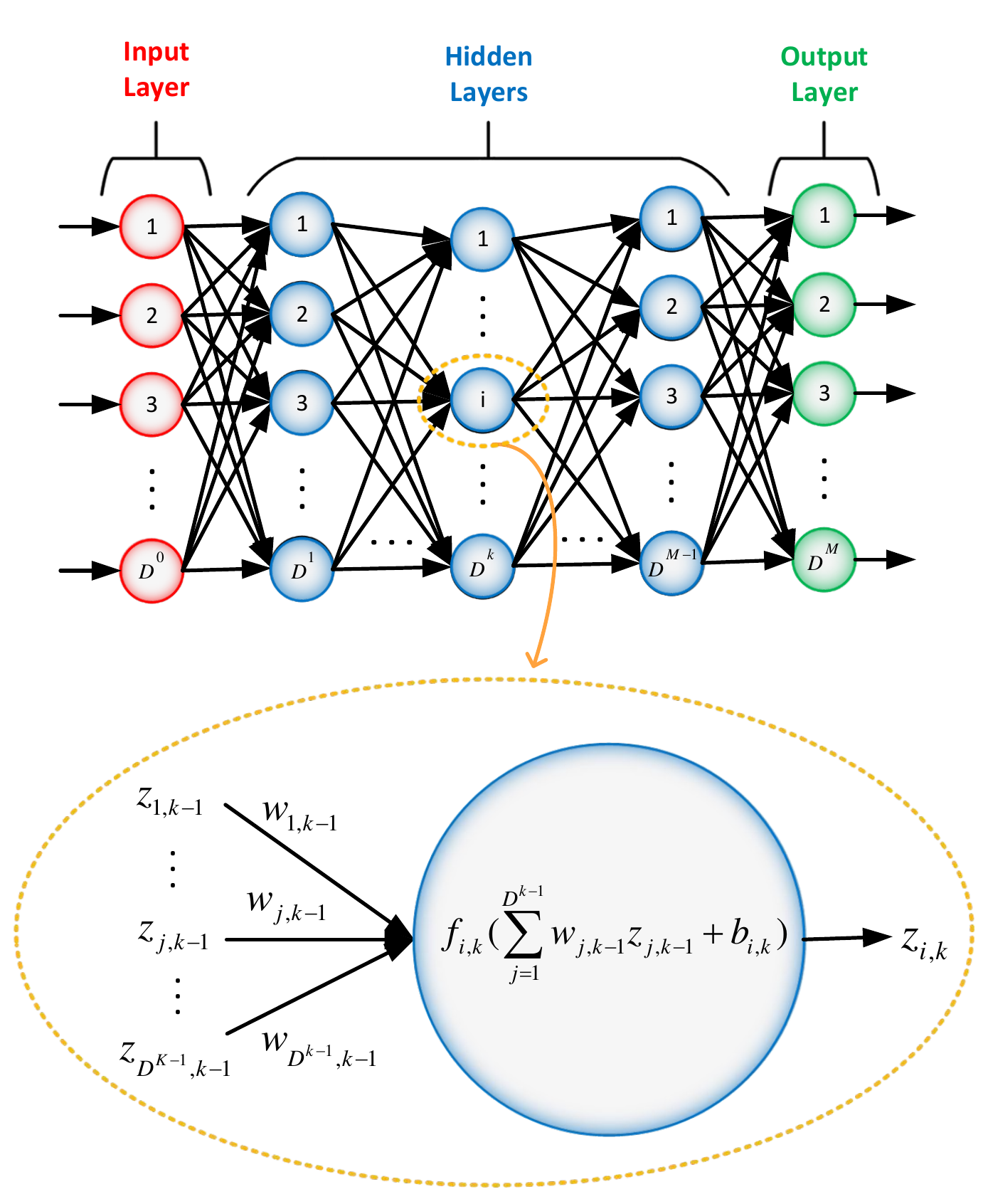}}
	\caption{Illustration of a feedforward fully connected DNN.}
	\label{fig:DNN}
\end{figure}

In the training process, weights and biases are updated iteratively using a supervised procedure that aims to minimize the error between the predicted and target values of the output layer. The loss function which calculates the error can be any suitable measure of the distance between the predicted and target values.
Once the training process converges and the weights and biases are configured, DNN can be used to predict the output for the inputs that have not been used in the training process. In fact, the DNN learns the mapping between the input and the output. In the sequel, we describe the detailed architecture and techniques of our proposed DNN-based method for Problem \ref{pr:power}.

\subsubsection{DNN architecture} \label{sec:DNNarc}
The power allocation problem, i.e., Problem \ref{pr:power}, can be considered as an unknown mapping $\mathscr{M}$ from the channel gains between a BS and its users (both connected and failed users) to the optimal power allocation for the failed and connected users in that BS, i.e.,


\begin{equation}
\mathscr{M}: \mathscr{H}=[\mathfrak{h}_{u',\ell'}] \mapsto \mathscr{P}^{*}=[\mathfrak{p}_{u',\ell'}^{*}]
\end{equation}
where superscript $*$ indicates the optimal values of the corresponding variable. $\mathfrak{h}_{u',\ell'}$ is equal to the channel gain of user $u'$ assigned to cluster $\ell' \in \mathcal{L}^n$ of BS $n$. $\mathfrak{p}_{u',\ell'}^{*}$ is the optimal power allocated to user $u'$ assigned to cluster $\ell' \in \mathcal{L}^n$ of BS $n$.

In order to emulate this mapping, the DNN is trained by a training set that includes $K_S$ pairs of the input and the target output, i.e., $\{ \mathscr{H}_{s}, \mathscr{P}^{*}_{s} \}_{s=1}^{K_S}$.

As discussed above, the DNN is composed of an input layer, multiple hidden layers, and one output layer. In our proposed DNN, activation functions $\emph{f}_{i,k}$ of the neurons in each layer are assumed to be ReLU which is defined as follows:
\begin{equation}
\emph{f}_{i,k}(X) =
  \begin{cases}
    X,       & \quad \text{if } X>0,\\
   0,  & \quad \text{otherwise}.\\
  \end{cases}
\end{equation}

\subsubsection{Dataset generation}
To build a dataset, we first position BSs in the network. Then, users are randomly distributed in the network area. Assuming that one BS experiences failure, failed users are assigned to the compensating cells based on Algorithm \ref{alg:userAssociation}. Then, Problem \ref{pr:power} is solved by an off-the-shelf interior point method solver for every cell that serves at least one failed user. This procedure is repeated until we have a large enough dataset, which is then divided into training, test, and validation datasets. We assume that we have $K_S$ pairs of $\{ \mathscr{H}_{s}, \mathscr{P}^{*}_{s} \}$ in our training dataset, and the rest are divided between test and validation datasets.

\subsubsection{Data preprocessing}
For better performance in the training phase, a preprocessing step is required to prepare the input and output data. We first express the channel gain values in logarithmic units to avoid numerical problems. We also notice that switching the columns and rows of the input (i.e., $\mathscr{H}$) does not affect the output when its columns are also switched accordingly. Therefore, we can produce new samples and augment the training data. For example, consider a base station with four connected users and two failed users assigned to two NOMA clusters, as shown in Figure \ref{fig:RowPer}. If the first and second rows of the matrix $\mathscr{H}$ are permutated (producing matrix $\mathscr{H'}$), then the same rows of the power allocation matrix $\mathscr{P}$ are also permutated accordingly (producing matrix $\mathscr{P'}$). The same property exists for column permutation as shown in Figure \ref{fig:ColPer}. The first and second columns of the matrix $\mathscr{H}$ are switched, and the same columns of the power allocation matrix are also switched.  As a result, a new sample is generated without the need to solve the optimization problem.
 \begin{figure}[!t]
	\centering
	\centerline{\includegraphics[width=\linewidth]{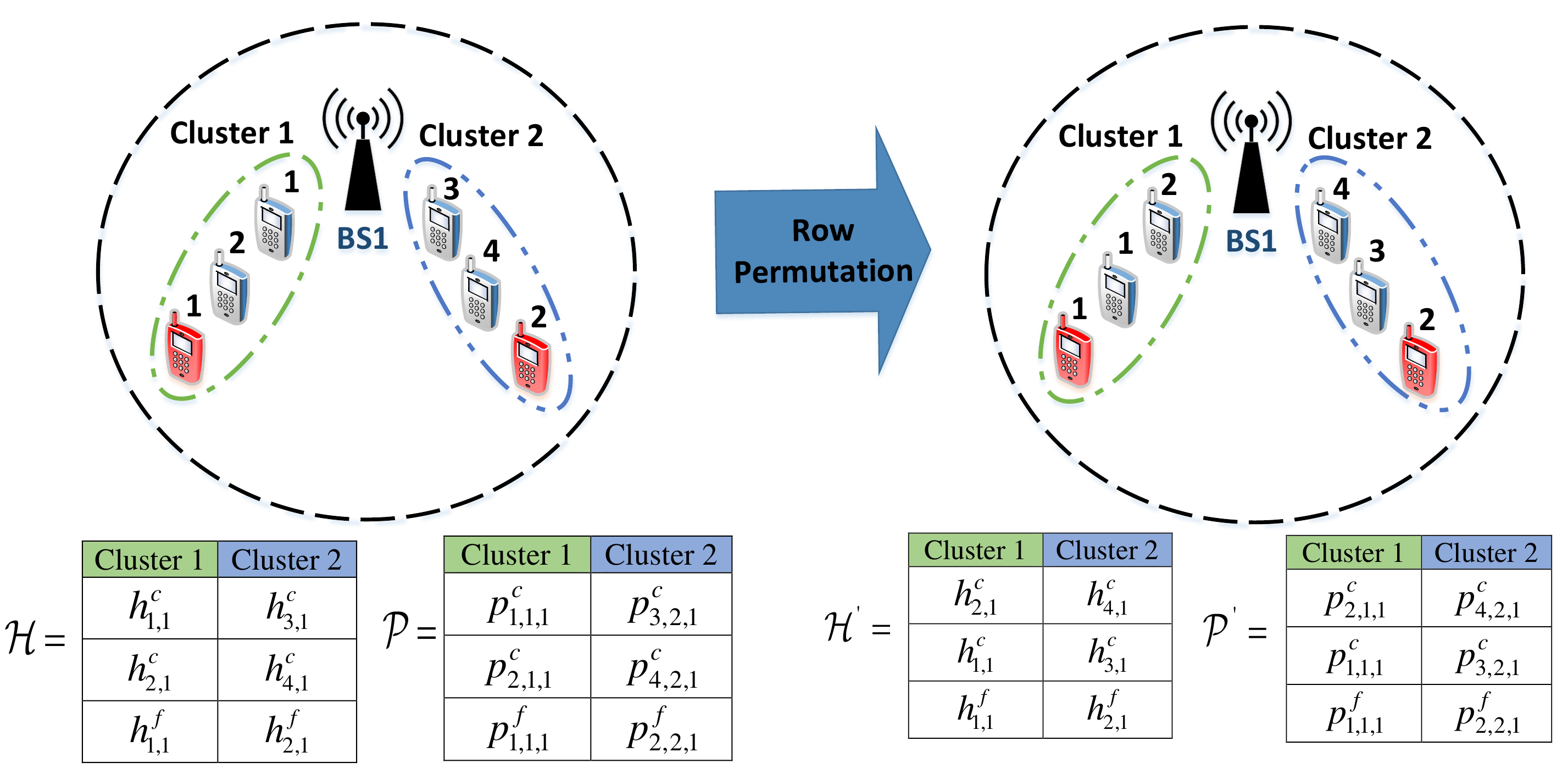}}
	\caption{Row permutation for both the input and output of the DNN.}
	\label{fig:RowPer}
\end{figure}

 \begin{figure}[!t]
	\centering
	\centerline{\includegraphics[width=\linewidth]{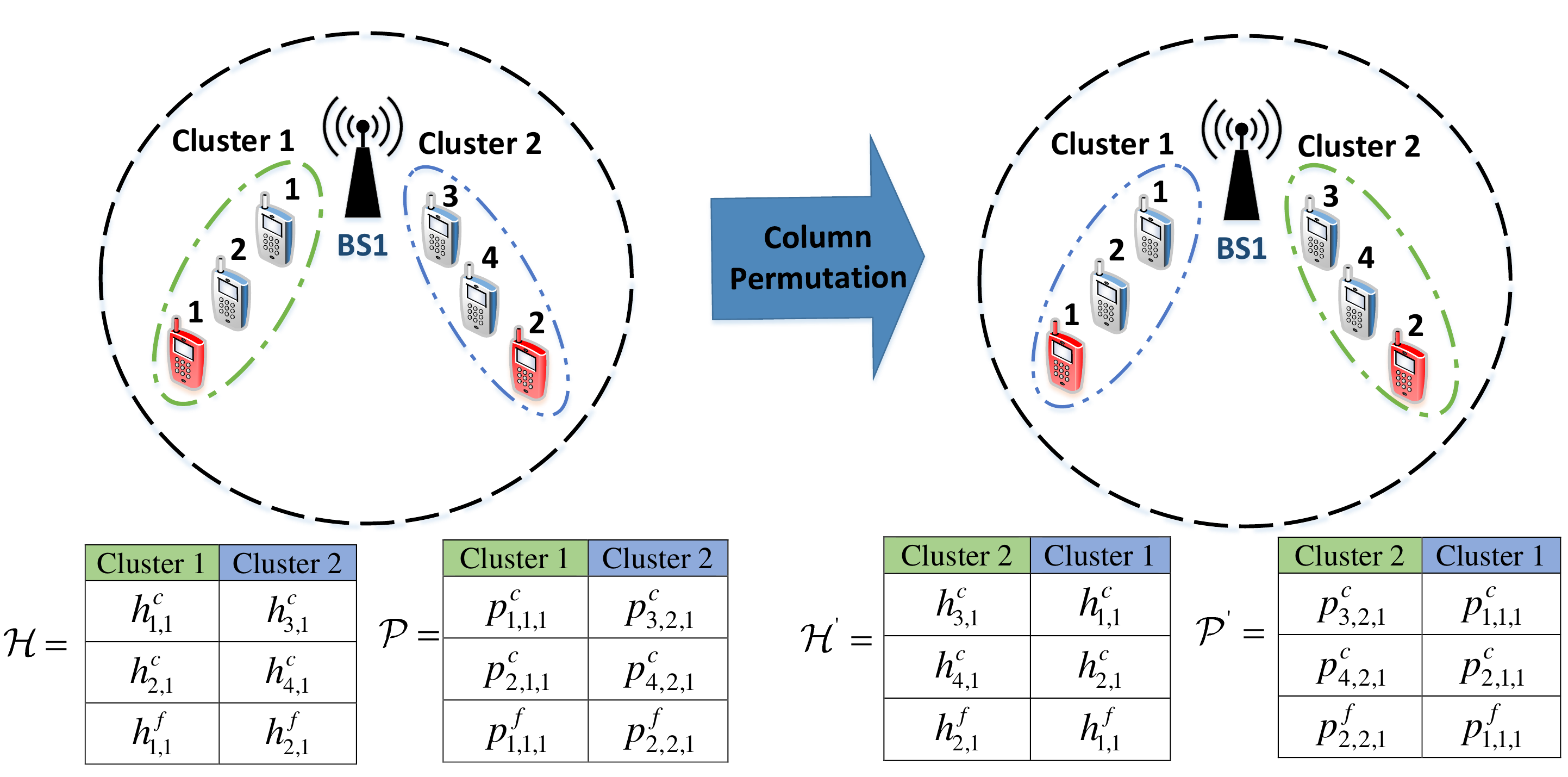}}
	\caption{Column permutation for both the input and output of the DNN.}
	\label{fig:ColPer}
\end{figure}

\subsubsection{Training procedure}
 First, weights and biases of each neuron in the DNN is randomly initialized. Let us denote the matrices of weights and biases of the whole DNN by  $\mathscr{W}$ and $\mathscr{B}$, respectively. In each iteration of the training process, a sample $\mathscr{H}_{s}$ from the training set is fed forward as input to the DNN, and predicted output $\hat{\mathfrak{p}}_{u',\ell'}$ is obtained. To calculate the error between the desired values and predicted values, we use the mean square error (MSE) function, which is a widely used loss function in regression problems \cite{goodfellow2016deep}. MSE is the average of squares of error between the predicted and optimal power values.

 The error is then backpropagated through the DNN to gradually adjust the weights and biases.
 We use ADAM optimizer with Nesterov momentum \cite{dozat2016incorporating}, which is a fast and powerful learning algorithm to update weights and biases according to the error.

 In this paper, we consider batch learning. In batch learning, $k$ samples are randomly chosen at each iteration, and the mean of the $k$ corresponding gradients is utilized in the updating procedure of weights and biases. When all samples in the entire dataset is passed through the DNN one time, an epoch is elapsed. If the batch size is $k$, $\frac{K_S}{k}$ number of parameter updates are performed during an epoch.

\subsubsection{Testing procedure} \label{sec:DNNtest}
To evaluate the performance and assess the predictive power and generalization of the DNN, we need a test dataset. The channel gains of the test dataset are passed through the network and predicted power values are obtained which are then compared against optimal power values.

\begin{remark}\label{rem:multipleBS}
Our proposed algorithm can be extended to the case of multi-cell outage experience in a straightforward manner. After the user assignment procedure, the proposed DNN approach for solving Problem \ref{pr:power} can be applied for any compensating cell with at least one failed user regardless of the number of failed cells.
\end{remark}
\section{Proposed solution in presence of co-channel interference}\label{sec:cochannelInt}
In this section, for more practicality, we assume that neighboring cells can operate on the same subchannel and cause co-channel interference to one another.
Let us denote by $\mathcal{B}$ the set of available subchannels. In addition, let $\zeta_{b,l,n}$ be the subchannel allocation parameter which equals to one if subchannel $b \in \mathcal{B}$ is allocated to the cluster $\ell \in \mathcal{L}^n$ of BS $n \in \mathcal{N}$. We follow the approach of \cite{abdelnasser2015resource} in taking the co-channel interference into account which is a commonly used approach in the literature. Our objective is to perform user association and power allocation in a way that can sustain the highest possible co-channel interference level. In this regard, a constraint is added to the problem limit the co-channel interference on each UE in a NOMA cluster to a maximum threshold. Then, we try to maximize the worst-case spectral efficiency of the failed users assuming that the co-channel interference constraint holds with equality. Please note that clusters in the same BS operate on different subchannels and do not cause interference to each other.

Problem \ref{pr:joint} in presence of co-channel interference is reformulated as follows.

\noindent\textbf{Problem 3.} Given the channel gains of failed and connected users in the network, clustering parameter $\beta_{u,\ell,n}$, and subchannel allocation parameter $\zeta_{b,l,n}$ for all $\forall b \in \mathcal{B}$, $n \in \mathcal{N}$, $\ell \in \mathcal{L}^n$ and $u \in \mathcal{U}^n$, the joint user association and power allocation optimization problem in the outage compensation scenario considering co-channel interference is formulated as follows:
\small
\begin{equation*}\label{eq:CointJointAfter}
\begin{split}
\displaystyle &\max_{\mathcal{A} , \mathcal{P}^{f}, \mathcal{P}^{c}}  \sum_{u \in \mathcal{U}^f} \sum_{n \in \mathcal{N}} \sum_{\ell \in \mathcal{L}^n} \alpha_{u,\ell,n} S_{u,\ell,n}^{}, \\
&\text{subject to:}\\
& \text{\textbf{C1}}: \log_2\bigg(1 + \frac{p^{c}_{u,\ell,n}h^{c}_{u,n}}{  h^{c}_{u,n}\sum\limits_{k =1}^{u-1} \beta_{k,\ell,n}p^{c}_{k,\ell,n}  + |\mathcal{D}_{n}| \times \mathfrak{I}^{\text{max}} + \sigma^2}\bigg) \geq \mathfrak{s}^{\text{min}}_{u},\\
& \hspace{30mm} \forall u \in \mathcal{U}^n, n \in \mathcal{N},\\
&\text{\textbf{C2\hspace{1mm}}}- \text{\textbf{C9\hspace{1mm}}}, \\
&\text{\textbf{C10}}: \\
& \zeta_{b,l',n'} \zeta_{b,l,n} \left( \sum_{u \in \mathcal{U}^f}  \alpha_{u,\ell,n}p_{u,\ell,n}^{f} + \sum_{u \in \mathcal{U}^n} \beta_{u,\ell,n} p^{c}_{u,\ell,n} \right) \times \alpha_{u',\ell,n'} h^{f}_{u,n}\\
  &\hspace{45mm}\leq \zeta_{b,l',n'} \zeta_{b,l,n} \mathfrak{I}^{\text{max}},\\
    &\hspace{12mm} \forall b \in \mathcal{B}, \forall n \in \mathcal{N}, \forall n' \in \mathcal{D}_{n}, \forall \ell \in \mathcal{L}^n, \forall \ell' \in \mathcal{L}^{n'}, \forall u' \in \mathcal{U}^f, \\
&\text{\textbf{C11}}: \\
& \zeta_{b,l',n'} \zeta_{b,l,n} \left( \sum_{u \in \mathcal{U}^f}  \alpha_{u,\ell,n}p_{u,\ell,n}^{f} + \sum_{u \in \mathcal{U}^n} \beta_{u,\ell,n} p^{c}_{u,\ell,n} \right) \times \beta_{u',\ell,n'} h^{c}_{u,n} \\
&\hspace{45mm} \leq \zeta_{b,l',n'} \zeta_{b,l,n} \mathfrak{I}^{\text{max}},\\
 &\hspace{12mm} \forall b \in \mathcal{B}, \forall n \in \mathcal{N}, \forall n' \in \mathcal{D}_{n}, \forall \ell \in \mathcal{L}^n, \forall \ell' \in \mathcal{L}^{n'}, \forall u' \in \mathcal{U}^c,
\end{split}
\end{equation*}
\normalsize
where $\mathcal{D}_{n}$ is the set of neighboring cells of BS $n$, and $\mathfrak{I}^{\text{max}}$ is the maximum co-channel interference that can be tolerated by failed and connected users in other neighboring cells, and
$$ S_{u,\ell,n}^{} = \log_2\bigg(1 + \frac{p^{f}_{u,\ell,n}h^{f}_{u,n}}{h^{f}_{u,n} \sum\limits_{k=1}^{U^n} \beta_{u,\ell,n} p^{c}_{k,\ell,n}  + |\mathcal{D}_{n}| \times \mathfrak{I}^{\text{max}} + \sigma^2}\bigg).$$
 Constraints \textbf{C10} and \textbf{C11} imply that the co-channel interference that one BS can generate for the failed and connected users in the neighboring cells is limited to $\mathfrak{I}^{\text{max}}$. Constraints \textbf{C10} and \textbf{C11} is active only if two clusters in a BS and its neighboring cell use the same subchannel.

In presence of co-channel interference, the user association algorithm is similar to the one described in Subsection \ref{sec:failedAssociation}, with the difference that in this case, in approximation of the power budget of each cluster which can be allocated to the failed users, we need to consider the co-channel interference effect. Similar to Section \ref{sec:PABefore}, we assume that prior to the outage event, BS $n \in \mathcal{N}$ optimal power allocation is done to maximize the sum of SE of its users while considering their minimum SE requirement as described in problem \eqref{eq:PABefor}. In co-channel interference scenario we have:
\small
\begin{equation}\label{eq:PABeforInt}
\begin{split}
 \displaystyle &\max_{\mathcal{P}^{'c}_n} \sum_{u \in \mathcal{U}^n} \sum_{\ell \in \mathcal{L}^n} \beta_{u,\ell,n} S_{u,\ell,n}^{'c},\\
&\text{subject to:}\\
&\text{\textbf{C1\hspace{1mm}:\hspace{1mm}}} \sum_{\ell \in \mathcal{L}^n} \beta_{u,\ell,n} S_{u,\ell,n}^{'c} \geq \mathfrak{s}^{\text{min}}_u, \forall u \in \mathcal{U}^n, \\
& \text{\textbf{C2\hspace{1mm}:\hspace{1mm}}} \big( \beta_{u,\ell,n} p^{'c}_{u,\ell,n} - \sum\limits_{k =1}^{u-1} \beta_{k,\ell,n} p^{'c}_{k,\ell,n} \big) H_{u-1,\ell,n}^{c}
 \geq\\
 &\hspace{20mm} \mathfrak{p}^{\text{tol}} - (1-\beta_{u,\ell,n})M, \hspace{2mm} 2 \leq u \leq U^n,\hspace{2mm}\forall \ell \in \mathcal{L}^n, \\
& \text{\textbf{C3\hspace{1mm}:\hspace{1mm}}} \sum_{\ell \in \mathcal{L}^n} \sum_{u \in \mathcal{U}^n} \beta_{u,\ell,n} p^{'c}_{u,\ell,n} \leq \mathfrak{p}^{\text{max}},\\
& \text{\textbf{C4\hspace{1mm}:\hspace{1mm}}} p^{'c}_{u,\ell,n} \geq 0, \hspace{2mm} \forall u \in \mathcal{U}^n, \hspace{2mm} \forall \ell \in \mathcal{L}^n,\\
&\text{\textbf{C5\hspace{1mm}:\hspace{1mm}}}: \zeta_{b,l',n'} \zeta_{b,l,n} \left(  \sum_{u \in \mathcal{U}^n} \beta_{u,\ell,n} p^{c}_{u,\ell,n} \right) \times \alpha_{u',\ell,n'} h^{f}_{u,n}\\
  &\hspace{47mm}\leq \zeta_{b,l',n'} \zeta_{b,l,n} \mathfrak{I}^{\text{max}},\\
    &\hspace{10mm} \forall b \in \mathcal{B}, \forall n \in \mathcal{N}, \forall n' \in \mathcal{D}_{n}, \forall \ell \in \mathcal{L}^n, \forall \ell' \in \mathcal{L}^{n'}, \forall u' \in \mathcal{U}^f \\
&\text{\textbf{C6\hspace{1mm}:\hspace{1mm}}}: \zeta_{b,l',n'} \zeta_{b,l,n} \left(  \sum_{u \in \mathcal{U}^n} \beta_{u,\ell,n} p^{c}_{u,\ell,n} \right) \times \beta_{u',\ell,n'} h^{c}_{u,n} \\
&\hspace{47mm} \leq \zeta_{b,l',n'} \zeta_{b,l,n} \mathfrak{I}^{\text{max}},\\
 &\hspace{10mm} \forall b \in \mathcal{B}, \forall n \in \mathcal{N}, \forall n' \in \mathcal{D}_{n}, \forall \ell \in \mathcal{L}^n, \forall \ell' \in \mathcal{L}^{n'}, \forall u' \in \mathcal{U}^c,
\end{split}
\end{equation}\normalsize
where
 $p^{'c}_{u,\ell,n}$ is the power that BS $n$ allocates to user $u$ which belongs to cluster $\ell \in \mathcal{L}^n$ at the time of failure, $\mathcal{P}^{'c}_n$ is a set whose elements are $p^{'c}_{u,\ell,n}$, and
  $$S_{u,\ell,n}^{'c} = \log_2\bigg(1 + \frac{p^{'c}_{u,\ell,n}h^{c}_{u,n}}{  h^{c}_{u,n}\sum\limits_{k =1}^{u-1} \beta_{k,\ell,n}p^{'c}_{k,\ell,n}  + \sigma^2 + |\mathcal{D}_{n}| \times \mathfrak{I}^{\text{max}}}\bigg).$$
Assume that $\mathcal{I}_{\ell, n} = \{i_1^{\ell , n}, i_2^{\ell , n}, \dots, i_q^{\ell , n} \}$ is the set of connected user indices in cluster $\ell$ of BS $n$ in a way that $h^{c}_{i_1^{\ell , n},n} \geq h^{c}_{i_2^{\ell , n},n} \geq \dots, h^{c}_{i_q^{\ell , n},n}$. From the constraint \textbf{C1} in Problem 3, it can be seen that the minimum amount of power that needs to be allocated to the user whose index is $i_1^{\ell ,n}$, can be determined on the basis of its QoS requirement as follows:
\begin{equation}\label{eq:minPowerMaxUserCoInt}
\begin{split}
  p_{i_1^{\ell ,n},\ell,n}^{c} &= \frac{\sigma^2 + |\mathcal{D}_{n}| \times \mathfrak{I}^{\text{max}}}{h_{i_1^{\ell ,n},n}^{c}}\left( 2^{\mathfrak{s}^{\text{min}}_{i_1^{\ell ,n}}} - 1\right), \\
  & \forall n \in \mathcal{N}, \hspace{2mm}\forall \ell \in \mathcal{L}^n.
\end{split}
\end{equation}
For the other connected users in cluster $\ell$ of BS $n$, the required amount of power can be determined according to the following proposition:

\noindent\textbf{Proposition 2}. Consider Problem 3 and \eqref{eq:PABeforInt}. With assumption that $\mathfrak{p}^{\text{tol}} << 1$, the amount of power that needs to be allocated to connected users after the outage event can be approximated as follows:
\begin{equation}\label{eq:appPowerConnectedInt}
\begin{split}
  p_{i_k^{\ell ,n},\ell,n}^{c} &\approx \frac{h_{i_k^{\ell ,n},n}\sum_{j = 1}^{k-1}p^{c}_{i_j^{\ell ,n},\ell,n} + \sigma^2 + |\mathcal{D}_{n}| \times \mathfrak{I}^{\text{max}}}{h_{i_k^{\ell ,n},n}\sum_{j = 1}^{k-1}p^{'c}_{i_j^{\ell ,n},\ell,n} + \sigma^2 + |\mathcal{D}_{n}| \times \mathfrak{I}^{\text{max}}}p_{i_k^{\ell ,n},\ell,n}^{'c},\\
   &\forall k \geq 2, n \in \mathcal{N}, \ell \in \mathcal{L}^n,
\end{split}
\end{equation}
where $ p_{i_1^{\ell ,n},\ell,n}^{c}$ can be calculated from \eqref{eq:minPowerMaxUserCoInt}.

\begin{proof}
If the power of connected users are determined based on \eqref{eq:appPowerConnectedInt},  the SE of connected users will be as follows:
\small
\begin{equation}\label{eq:SEAfetrInt}
\begin{split}
  S_{i_k^{\ell ,n},\ell,n}^{c} &= \log_2\bigg(1 + \frac{h^{c}_{i_{k^{\ell ,n},n}}p_{i_k^{\ell ,n},\ell,n}^{c}}{  h^{c}_{i_{k^{\ell ,n},n}}\sum\limits_{j =1}^{k-1} p^{c}_{k,\ell,n}  + \sigma^2 + |\mathcal{D}_{n}| \times \mathfrak{I}^{\text{max}}}\bigg)\\
  & = \log_2\bigg(1 + \frac{h^{c}_{i_{k^{\ell ,n},n}}p_{i_k^{\ell ,n},\ell,n}^{'c}}{h_{i_k^{\ell ,n},n}\sum_{j = 1}^{k-1}p^{'c}_{i_j^{\ell ,n},\ell,n} + \sigma^2 + |\mathcal{D}_{n}| \times \mathfrak{I}^{\text{max}}}\bigg)\\
 & = S_{i_{k^{\ell ,n}},\ell,n}^{'c}
\end{split}
\end{equation}
\normalsize
Since SE requirement for connected users is the same before  and after the outage event, from \eqref{eq:SEAfetrInt} we can conclude that the power allocation according to \eqref{eq:appPowerConnectedInt} satisfies the QoS requirement of the connected users.

\noindent Based on power allocation in \eqref{eq:minPowerMaxUserCoInt}, we have
\small
\begin{equation}\label{eq:proffCons}
\begin{split}
 p_{i_k^{\ell ,n},\ell,n}^{c}& - \sum_{j=1}^{k-1}p_{i_j^{\ell ,n},\ell,n}^{c}\\
 &= \frac{h_{i_k^{\ell ,n},n}\sum_{j = 1}^{k-1}p^{c}_{i_j^{\ell ,n},\ell,n} + \sigma^2 + |\mathcal{D}_{n}| \times \mathfrak{I}^{\text{max}}}{h_{i_k^{\ell ,n},n}\sum_{j = 1}^{k-1}p^{'c}_{i_j^{\ell ,n},\ell,n} + \sigma^2 + |\mathcal{D}_{n}| \times \mathfrak{I}^{\text{max}}}p_{i_k^{\ell ,n},\ell,n}^{'c}\\
 & \hspace{10mm} - \sum_{j=1}^{k-1}p_{i_j^{\ell ,n},\ell,n}^{c}\\
 &\stackrel{(a)}{\geq} \frac{h_{i_k^{\ell ,n},n}\sum_{j = 1}^{k-1}p^{c}_{i_j^{\ell ,n},\ell,n} + \sigma^2 + |\mathcal{D}_{n}| \times \mathfrak{I}^{\text{max}}}{h_{i_k^{\ell ,n},n}\sum_{j = 1}^{k-1}p^{'c}_{i_j^{\ell ,n},\ell,n} + \sigma^2 + |\mathcal{D}_{n}| \times \mathfrak{I}^{\text{max}}}\sum_{j=1}^{k-1}p_{i_j^{\ell ,n},\ell,n}^{'c} \\
 & \hspace{10mm}- \sum_{j=1}^{k-1}p_{i_j^{\ell ,n},\ell,n}^{c}\\
 & = \frac{(\sigma^2 + |\mathcal{D}_{n}| \times \mathfrak{I}^{\text{max}})(\sum_{j=1}^{k-1}p_{i_j^{\ell ,n},\ell,n}^{'c} - \sum_{j=1}^{k-1}p_{i_j^{\ell ,n},\ell,n}^{c})}{h_{i_k^{\ell ,n},n}\sum_{j = 1}^{k-1}p^{'c}_{i_j^{\ell ,n},\ell,n} + \sigma^2 + |\mathcal{D}_{n}| \times \mathfrak{I}^{\text{max}}}\\
 &\geq 0,
 \end{split}
\end{equation}
\normalsize
where $(a)$ follows from \textbf{C2} in \eqref{eq:PABeforInt}. with the assumption of $\mathfrak{p}^{\text{tol}} <<  1$, from \eqref{eq:proffCons} we can conclude that power approximation as \eqref{eq:appPowerConnectedInt} meets constraint \textbf{C2} in \eqref{eq:CointJointAfter}. Finally, due to the fact that co-channel interference threshold is the same before and after the outage, the power approximation in \eqref{eq:appPowerConnectedInt} does not deviate co-channel interference constraint which completes the proof.
\end{proof}

Having the required power of connected users at hand, we can determine the amount of power each cluster can allocate to a failed user as follows:
\begin{equation}\label{eq:culsterPowerBudgetCo}
\begin{split}
 \Delta_{\ell,n} &= \sum_{u \in \mathcal{U}^n}\beta_{u,\ell,n} p^{'c}_{u,\ell,n} - \sum_{u \in \mathcal{U}^n}\beta_{u,\ell,n} p^{c}_{u,\ell,n} ,\\
  &\hspace{15mm} \forall n \in \mathcal{N}, \hspace{2mm} \forall n \in \mathcal{N}, \ell \in \mathcal{L}^n,
 \end{split}
 \end{equation}
where $p^{c}_{u,\ell,n}$ are calculated based on \eqref{eq:minPowerMaxUserCoInt} and \eqref{eq:appPowerConnectedInt}.
Now from \eqref{eq:culsterPowerBudgetCo},  we can propose an efficient scheme for the association of failed users to the clusters of compensating cells as described in Algorithm \ref{alg:userAssociationCo} where
\small
\begin{equation}\label{eq:failedUserRateCo}
       \begin{split}
   S^{f}_{u,\ell,n} &=
    \log_2\bigg(1 + \frac{\Delta_{\ell,n}h^{f}_{u,n}}{  h^{f}_{u,n}\sum\limits_{k =1}^{U^n} \beta_{u,\ell,n} \delta_{\ell,n}p^{'c}_{k,\ell,n}  + |\mathcal{D}_{n}| \times \mathfrak{I}^{\text{max}} + \sigma^2}\bigg),\\
     &\hspace{10mm} \forall n \in \mathcal{N}, \hspace{2mm} \forall \ell \in \mathcal{L}^n, \hspace{2mm} \forall u\in U^f.
     \end{split}
 \end{equation}\normalsize
\begin{algorithm}[t]
	\caption{: Heuristic user association algorithm in presence of co-channel interference}\label{alg:userAssociationCo}
	Initialization: $\alpha_{u,\ell,n} = 0, \hspace{2mm} \forall u \in \mathcal{U}^f, n \in \mathcal{N}, \ell \in \mathcal{L}^n$.
	\begin{algorithmic}
		\ForAll{$n \in \mathcal{N}, \ell \in \mathcal{L}^n, u \in \mathcal{U}^f$ }
		\State Calculate $S^{f}_{u,\ell,n}$ using \eqref{eq:failedUserRateCo};
		\EndFor
        \ForAll{$i = 1, 2, \dots, U^f$}
        \State Determine $(u^*_i,\ell^*_i,n^*_i)$ according to \eqref{eq:trp};
        \State $\alpha_{u^*_i,\ell^*_i,n^*_i} = 1$;
        \State Update $S^{f}_{u,\ell,n}$ using \eqref{eq:SEUpdate1}-\eqref{eq:SEUpdate2};
        \EndFor
	\end{algorithmic}
\end{algorithm}

Given the user association scheme, i.e., $\mathcal{A}$, and the fact that constraints \textbf{C10} and \textbf{C11} are linear, Problem 3 remains convex. We follow the same procedures as in Subsections IV.B.1-5 to build and train our DNN.
The only difference is that the input to the DNN not only consists of channel gains between a BS and its users, but also the channel gains between a BS and users in the neighboring cells that operates in the same subchannel. Therefore, the new input layer has more number of neurons compared to the case where co-channel interference does not exist. To reduce the number of neurons in the input layer, we leverage the fact that when the co-channel interference that a BS generates on the neighboring user with the highest channel gain is lower than the threshold, the co-channel interference it generates on other neighboring users in the same subchannel also satisfies the constraint.

\section{Extend to Multi-Antenna Case}\label{sec:MIMO}
The multiple antennas at the transmitter (and/or receiver) side provide additional degrees of freedom for performance improvement \cite{7244171}. Hence, the application of MIMO in NOMA is of great interest and is well-discussed in the literature \cite{7236924}, \cite{7433470}. The main goal of this paper is proposing a method for cell outage compensation in NOMA scheme. Thus, without loss of generality and for simplicity, we considered SISO-NOMA. To complete the argument, in this section, we explain how in three steps our proposed scheme can be extended to MIMO-NOMA case.
\subsubsection{Model Adaptation}
First, let discuss how the problem formulation will change in MIMO-NOMA case. Consider the same network architecture as depicted in Figure \ref{fig:SystemModel} with the base stations, each of which is equipped with $N_{T_x}$ antennas at the transmitter side. The users are also equipped with $M_{R_x}$ antennas at the receiver side. Adopting the proposed framework in \cite{9093213}, we assume that each BS serves each of its clusters with one of its antennas. Therefore, the maximum number of clusters in each BS would be $N_{T_x}$. It should be noted that, with this assumption, we have no change at transmitter side on BS. However, we need to adapt the receive model at the users' side to MIMO.
Let $\textbf{y}_{u,n}$ denote the received signal at the $u$th user connected to BS $n$. It can be expressed as

\begin{equation}\label{eq:RxSig-MIMOcase}
  \text{\textbf{y}}_{u,n} = \text{\textbf{H}}_{u,n}\text{\textbf{P}}_{n}\text{\textbf{s}}_n+\text{\textbf{z}}_u,
\end{equation}
where $\textbf{H}_{u,n}\in\mathbb{C}^{M_{R_x}\times{N_{T_x}}}$ and $\textbf{P}_{n}\in{\mathbb{C}^{N_{T_x}\times{N_{T_x}}}}$ are the channel matrix between user $u$ and BS $n$ and the precoder matrix of BS $n$, respectively; $\text{\textbf{s}}_{n}\in{\mathbb{C}^{N_{T_x}\times1}}$ with the following definition is the transmitted vector containing the intended symbols of all clusters, and $\text{\textbf{z}}_{u}\in\mathbb{C}^{M_{R_x}\times1}$ is the noise vector,
\begin{equation}\label{eq:precoder-MIMOcase}
\begin{split}
  \text{\textbf{s}}_n = &\textit{\textbf{vec}}\left(m_l\right)_{l=1}^{T_x},\hspace{10.0mm}m_l = \sum_{u=1}^{U^n}{p_{l,u}}\times{s_{l,u}},\\
  &\forall n\in\mathcal{N},\hspace{10.0mm}\forall l\in\mathcal{L}^n,
\end{split}
\end{equation}
where $m_l$ is the aggregated transmitted signal to users in cluster $l$, $p_{l,u}$ is the power allocated to users $u$ of cluster $l$, and, $s_{l,u}$ is the transmitted symbol to $u$th user of cluster $l$ with normalized power.
As it can be seen from \eqref{eq:RxSig-MIMOcase}, each user receives a combination of all transmitted signals by $M_{R_x}$ antennas. It is well-known in literature that the Maximal-Ratio Combining (MRC) with appropriate design of precoder matrix $\text{\textbf{P}}_n$ and its corresponding detector at the receiver side is optimal \cite{7244171}. Using \eqref{eq:RxSig-MIMOcase}, \eqref{eq:precoder-MIMOcase} and mathematically rephrasing our analysis above, we will have the extended version of the problem formulation in MIMO case.

\subsubsection{Compensation in MIMO-NOMA case}
In the second step, we will show that how our proposed solution will be applicable to the cell outage compensation problem in MIMO-NOMA scenario. The proposed user association algorithm is based on the answer to this question that
how large the SE of a failed user would be if it is served by a candidate compensatory cell, knowing the extra amount of power of each cell while being able to fulfill the QoS requirements of its connected users. We have the similar arguments for the user association algorithm in MIMO case with no changes.
For the NN-based power allocation method, in MIMO-NOMA case, the input variables of the NN of cell $n$ would be the channel matrices and the precoder matrix ${\text{\textbf{H}}_{u,n}}$ and $\textbf{P}_n$, $\forall{u\in\mathcal{U}^n},\forall{n\in\mathcal{N}}$, respectively, and the output is the power allocation coefficients $p_{l,u}, \forall{l\in\mathcal{L}^n},\forall{u\in\mathcal{U}^n}$.

\subsubsection{Model Training}
The third step toward developing the proposed scheme for MIMO-NOMA is the training process of the neural network. For this purpose, we adopt the two-step training method proposed in \cite{9093213}. In this method, the network is initially trained based on an AWGN channel and using water-filling algorithm and then it is regarded as a basic model that can be trained for representing another channel type using transfer learning \cite{9336290}. In the secondary phase and in order to adapt to new environment samples, an online learning policy for power allocation is proposed.\\
Using the three steps explained above, our proposed method for cell compensation in SISO-NOMA and its corresponding analysis can be extended to MIMO-NOMA case.

\section{Performance Evaluation}\label{sec:simulation}
In this section, we first present the complexity analysis of our proposed method and that of the mathematical-based approach. Then, the simulation results are provided in the next subsection.
\subsection{Complexity Analysis} \label{sec:complexity}
Here, we analyze the computational complexity of the proposed suboptimal approach for Problem \ref{pr:joint} and compare it with that of the mathematical-based solution method. To obtain the global optimal solution, we need to examine all possible user association schemes and solve the corresponding convex optimization problem (Problem \ref{pr:power}) for all BSs that serve at least one failed user. Let $L$ be the total number of clusters in the compensating cells, i.e., $L = \sum_{n \in \mathcal{N}} L^n$. The number of possible association schemes is $\mathbb{P}(L,U^f)$ where $\mathbb{P}(x,y)$ is the number of $y$ permutations of $x$ objects.
Problem \ref{pr:power} is convex and can be solved using the interior point method. The worst-case complexity of the interior point method for this problem is $O((U^f + \sum_{n \in \mathcal{N}} U^n)^3)$ \cite{boyd2004convex}.
Therefore, the overall complexity of finding the global optimal solution is $O(\mathbb{P}(L,U^f) (U^f + \sum_{n \in \mathcal{N}} U^n)^3)$. It should be noted that the optimization-based method is completely performed in real time and only suffers from online complexity. As mentioned above, this makes the method inapplicable for cell outage compensation, which is inherently time sensitive.

Our proposed method includes two sequential steps: user association and DNN-based power allocation. The proposed failed user association algorithm consists of computing and sorting $L\times U^f$ numbers. Hence, its order of complexity is $O(L\times U^f \log(L\times U^f))$.
The computational complexity of the DNN-based solution method can be divided into the following two stages:
\subsubsection{Offline complexity}
Offline computational complexity involves generating the labeled dataset and training the DNN. Generating each labeled sample of the dataset requires solving Problem \ref{pr:power} for a realization of the channel gains with worst-case complexity $O((U^f + \sum_{n \in \mathcal{N}} U^n)^3)$. Assuming that we have $K_S$ number of samples in the dataset, the overall complexity of this step is $ O(K_S(U^f + \sum\limits_{n \in \mathcal{N}} U^n)^3)$.
\subsubsection{Online complexity}
Assume that the DNN has $M$ layers (including hidden and output layers), each layer $m$ having $D^{m}$ number of neurons. Given any input, the output of the DNN can be obtained by calculating $\sum_{m=1}^{M} D^{m-1} D^{m}$ number of multiplications, $\sum_{m=1}^{M} D^{m}$ number of summations, and $\sum_{m=1}^{M} D^{m}$ number of calculations of activation functions. In the online step, the DNN architecture is fixed and the number of calculations is constant. Therefore, the complexity of DNN is $O(1)$. Since the DNN needs to be performed for $N$ number of compensating cells, the overall complexity of this step is $O(N)$.

To summarize, the optimization-based method has worst-case exponential complexity, while the complexity of our proposed method is in the order of polynomial basis. It should also be pointed out that offline computational complexity can be afforded, while real-time developments have constrained resources.
\subsection{Simulation Results} \label{sec:results}
In this subsection, we evaluate the performance of the proposed solution through simulations. We consider a network consisting of one BS, which experiences failure, and its compensating cells. The cells radius is $120$ m. Users are randomly and uniformly placed in each cell where the distance between each user and BS is at least $20$ m. The path loss at a distance $d$ m from the BS is modeled as  $38 + 30\text{log}_{10}(d)$ dB \cite{ituProp}. The simulation setup parameters are summarized in Table \ref{tab:simulationSetup}.
\begin{table}[t]
  \centering
  \caption{Simulation Parameters}
  \label{tab:simulationSetup}
  \begin{tabular}{|| l | l ||}
  \hline
  \textbf{Parameter} & \textbf{Value}  \\
  \hline \hline
  $\mathfrak{p}^{\text{max}}$ & 46.02 dBm \\ \hline
  $\mathfrak{p}^{\text{tol}}$ & -101.4 dBm \cite{ali2016dynamic}\\ \hline
  $\sigma^2$  & -150 dBm\\ \hline
  $\mathfrak{s}^{\text{min}}_u$ & 4 bit/sec/Hz\\ \hline
  Path loss model & $38 + 30\text{log}_{10} d$ dB \cite{ituProp}\\
  \hline
\end{tabular}
\end{table}
For the DNN-based power control, we use a DNN with one input layer, three hidden layers, and one output layer. Each hidden layer has 200 neurons. Input to the DNN are the channel coefficients of users (connected and failed) in one BS. The output of the network is the power that has been assigned to each user and its size is the same as the input. The proposed DNN is implemented in Python 3.2.2 using Keras library \cite{chollet2018keras} with tensorflow \cite{abadi2016tensorflow} as back-end. We perform cross-validation to find the best hyperparameters for the training process. Unless otherwise noted, the parameters selected for our DNN are shown in Table \ref{tab:DNNparam}.

To generate the dataset, we use an off-the-shelf solver for Problem \ref{pr:power} in MATLAB. Moreover, we exploit the permutation invariance property of the input data in order to augment the training dataset and reduce the time of generating the labeled dataset.

For the sake of brevity, we use the following abbreviations to refer to the different approaches:
\begin{itemize}
  \item LC\_NOC refers to the proposed method, which is a low complexity algorithm for the NOMA-based outage compensation. This is fully discussed in Section \ref{sec:solution}.
  \item OPT\_NOC corresponds to the global optimal solution of Problem \ref{eq:JointAfter}.
  \item DNN\_PA refers to the proposed DNN for the power allocation problem discussed in Subsection \ref{sec:PA}.
  \item OPT\_PA refers to the optimal solution of power allocation Problem \ref{pr:power}.
  \item No\_OC corresponds to the case where no compensation technique is used after the cell outage.
\end{itemize}
\begin{table}[!t]
\caption{Parameters of DNN\_PA}
\label{param}
\centering
 \begin{tabular}{|| l | l ||}
 \hline
\textbf{Parameter} & \textbf{Value}\\ [0.5ex]
 \hline\hline
     Batch Size    &  128              \\ \hline
     Learning Rate &   0.0005           \\ \hline
     Decay Rate    &   0.9  \\ \hline
     Number of samples in whole dataset     &   130,000         \\ \hline
     Percentage of whole samples in training set     &   70\%          \\ \hline
     Percentage of whole samples in validation set     &   15\%         \\ \hline
     Percentage of whole samples in test set     &   15\%     \\ \hline
    \end{tabular}
    \label{tab:DNNparam}
\end{table}

\subsubsection{Effect of the batch size and learning rate}
In this experiment, we illustrate the effect of batch size and learning rate on the mean square error of DNN\_PA for the validation dataset. As we can see from Figure \ref{fig:batch}, smaller batch sizes achieve better stability and generalization performance, leading to lower MSE for the validation dataset. However, smaller batch sizes incur a greater computational cost since the optimization should be performed for more iterations during an epoch. We choose a batch size of 128 for our DNN to address this trade-off. For this batch size, we vary the learning rate to evaluate its effect on convergence. As Figure \ref{fig:LR} shows, large learning rates result in divergent behavior and small learning rates require more steps to converge. In fact, when the learning rate is small, weights and biases are updated accordingly and the loss decreases at a shallow rate. However, increasing the learning rate may cause the learning to jump over the minimum. For our DNN, a learning rate of 0.0005 is selected to address this trade-off.
 \begin{figure}[!t]
	\centering
	\centerline{\includegraphics[width=0.9\linewidth]{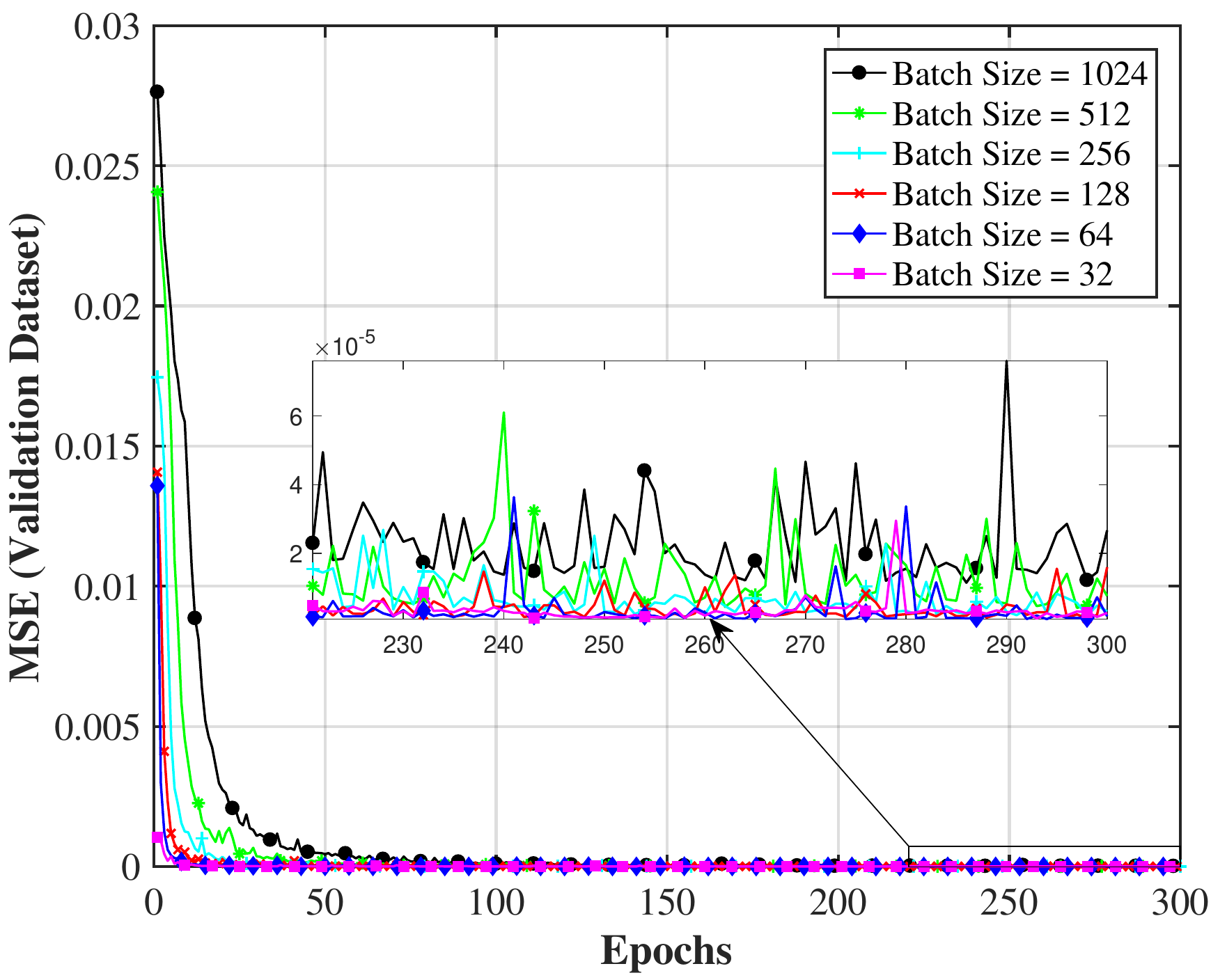}}
	\caption{Effect of the batch size on the mean square error.}
	\label{fig:batch}
\end{figure}
 \begin{figure}[!t]
	\centering
	\centerline{\includegraphics[width=0.9\linewidth]{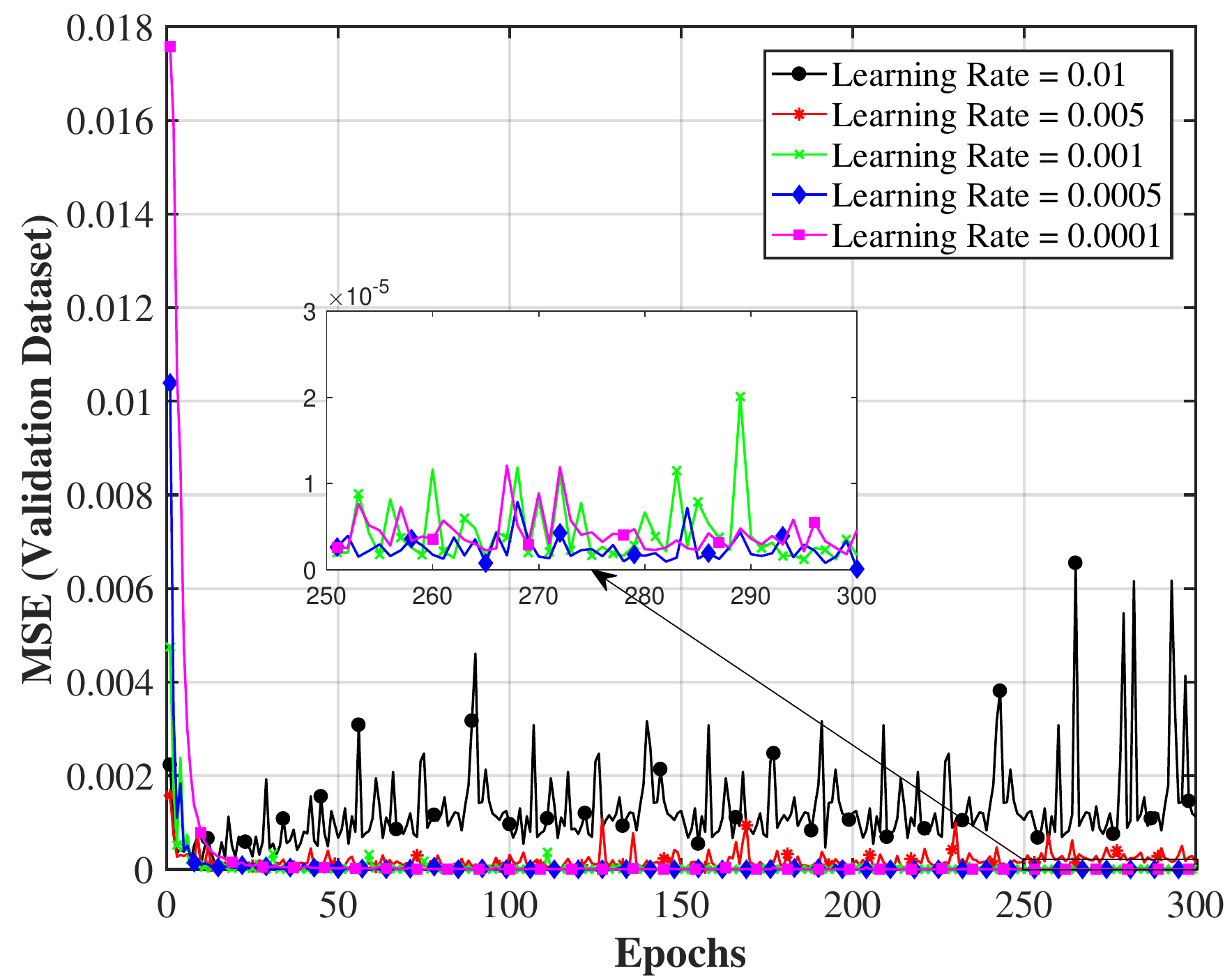}}
	\caption{Effect of the learning rate on the mean square error.}
    \label{fig:LR}
\end{figure}

\subsubsection{Performance evaluation of our proposed method}
In this experiment, we investigate the average SE of failed users in the solution of LC\_NOC and compare the results with those of OPT\_NOC. Since solving Problem \ref{eq:JointAfter} is very time consuming, we limit the number of compensating cells to $N=2$ and $N=3$ in this experiment, each of which serves  $U^n = 4$ connected users. The average SE of failed users is reported in Figure \ref{fig:comparison}. As we can see, the average SE of LC\_NOC is very close to that of the optimal solution. This suggests that our assumptions in the proposed algorithms are not restrictive and that the DNN is well trained to allocate near optimal power to compensating BSs.
\begin{figure}[!t]
	\centering
	\centerline{\includegraphics[width=0.95\linewidth]{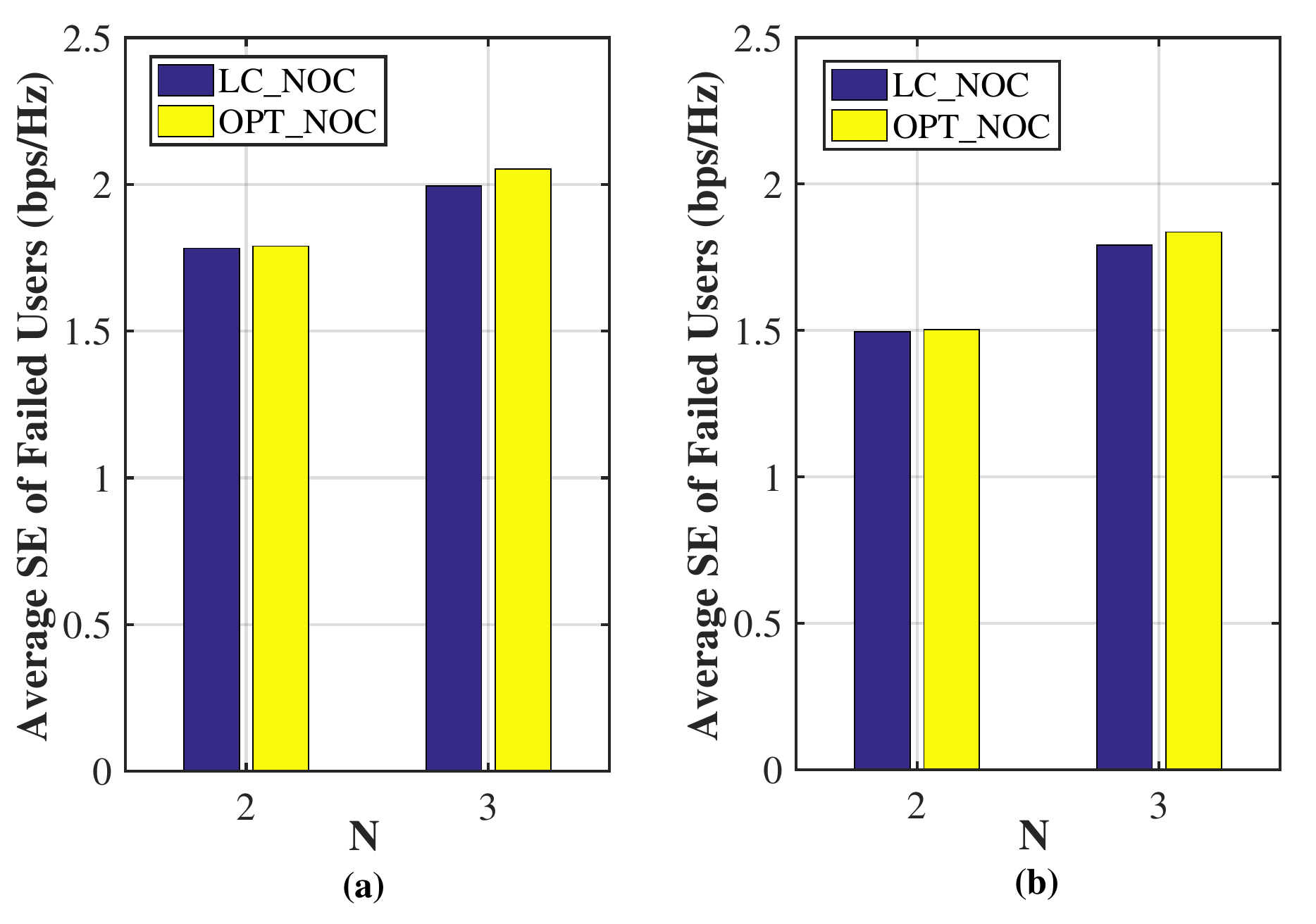}}
	\caption{Average SE of the failed users of LC\_NOC method compared to the OPT\_NOC method, (a): $U^f=3$, (b): $U^f=4$.}
	\label{fig:comparison}
\end{figure}


\subsubsection{Fairness evaluation of our proposed method}
In this experiment, LC\_NOC is compared to No\_OC. We vary the number of failed users and present values of the Jain's fairness (JF) index and average SE of all users, including both connected and failed users in Figure \ref{fig:noCompensation}. Jain's fairness index is defined as follows \cite{jain1984quantitative}:\small
\begin{equation*}
\displaystyle \text{JF} = \frac{(\sum\limits_{n \in \mathcal{N}}\sum\limits_{u \in \mathcal{U}^n} \text{SE}_{u,n}^c + \sum\limits_{n \in \mathcal{N}}\sum\limits_{u \in \mathcal{U}^f} \text{SE}_{u,n}^f )^2}{(U^f + \sum\limits_{n \in \mathcal{N}}U^n) \bigg(\sum\limits_{n \in \mathcal{N}}\sum\limits_{u \in \mathcal{U}^n} (\text{SE}_{u,n}^c)^2 + \sum\limits_{n \in \mathcal{N}}\sum\limits_{u \in \mathcal{U}^f} (\text{SE}_{u,n}^f)^2\bigg) },
\end{equation*}\normalsize
where $\text{SE}_{u,n}^c $and $\text{SE}_{u,n}^f$ are the achieved SE of connected user $u$ and failed user $u$ if it is served by BS $n$, respectively. If user $u$ is not served by any BS, its corresponding achieved SE is zero. As Figure \ref{fig:fairness} shows, LC\_NOC provides much higher fairness. Besides, using LC\_NOC, a greater number of users in the network are served. As expected, Figure \ref{fig:AverageSE} shows that the average SE of LC\_NOC is lower due to the fact that a portion of resources of compensating cells is now allocated to the failed users with less channel gain. However, we observe that the reduction in average SE is at most 5\%.

\begin{figure}
    \centering
    \subfigure[]
    {
        \includegraphics[width=0.9\linewidth]{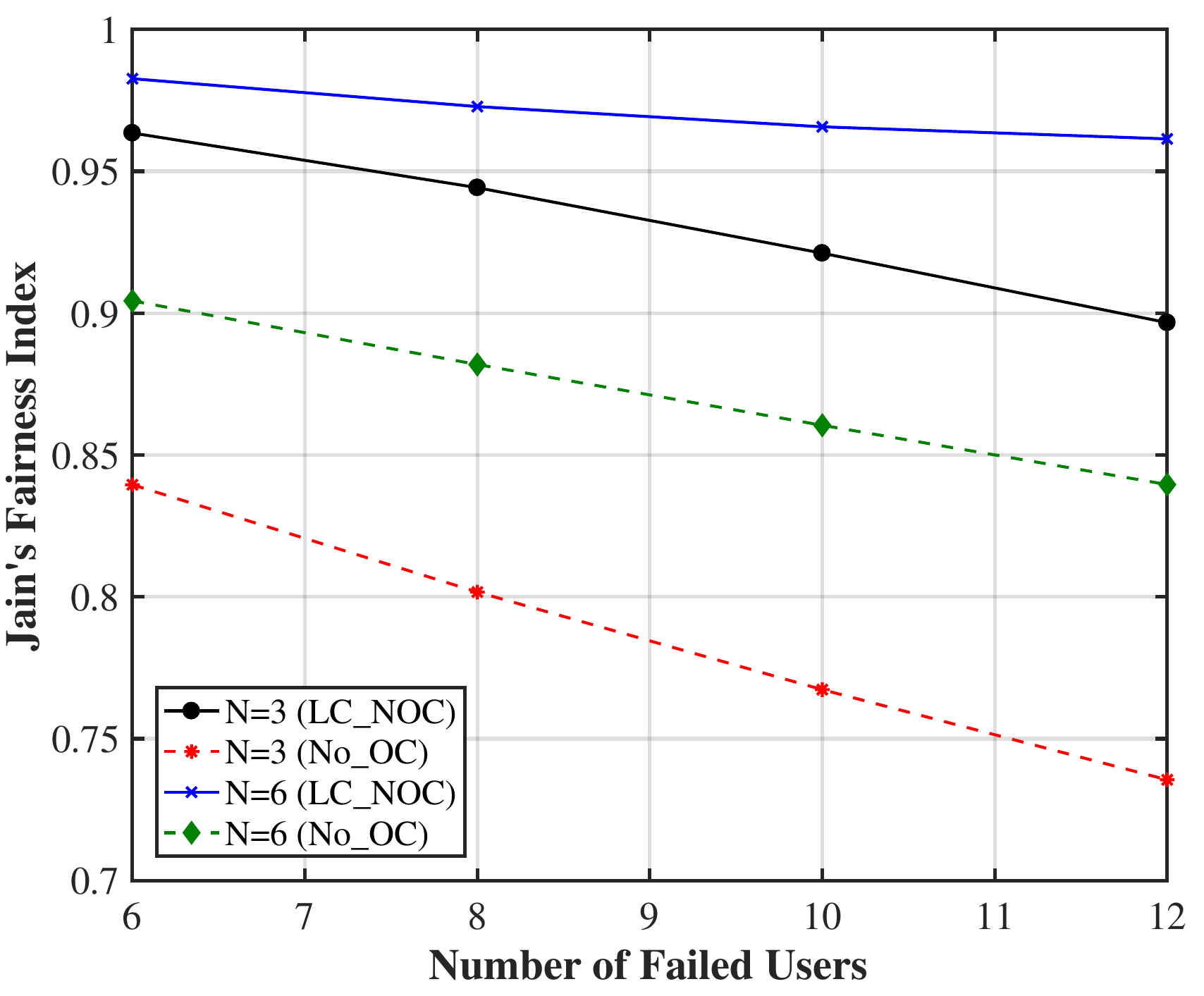}
        \label{fig:fairness}
    }
    \\
    \subfigure[]
    {
        \includegraphics[width=0.9\linewidth]{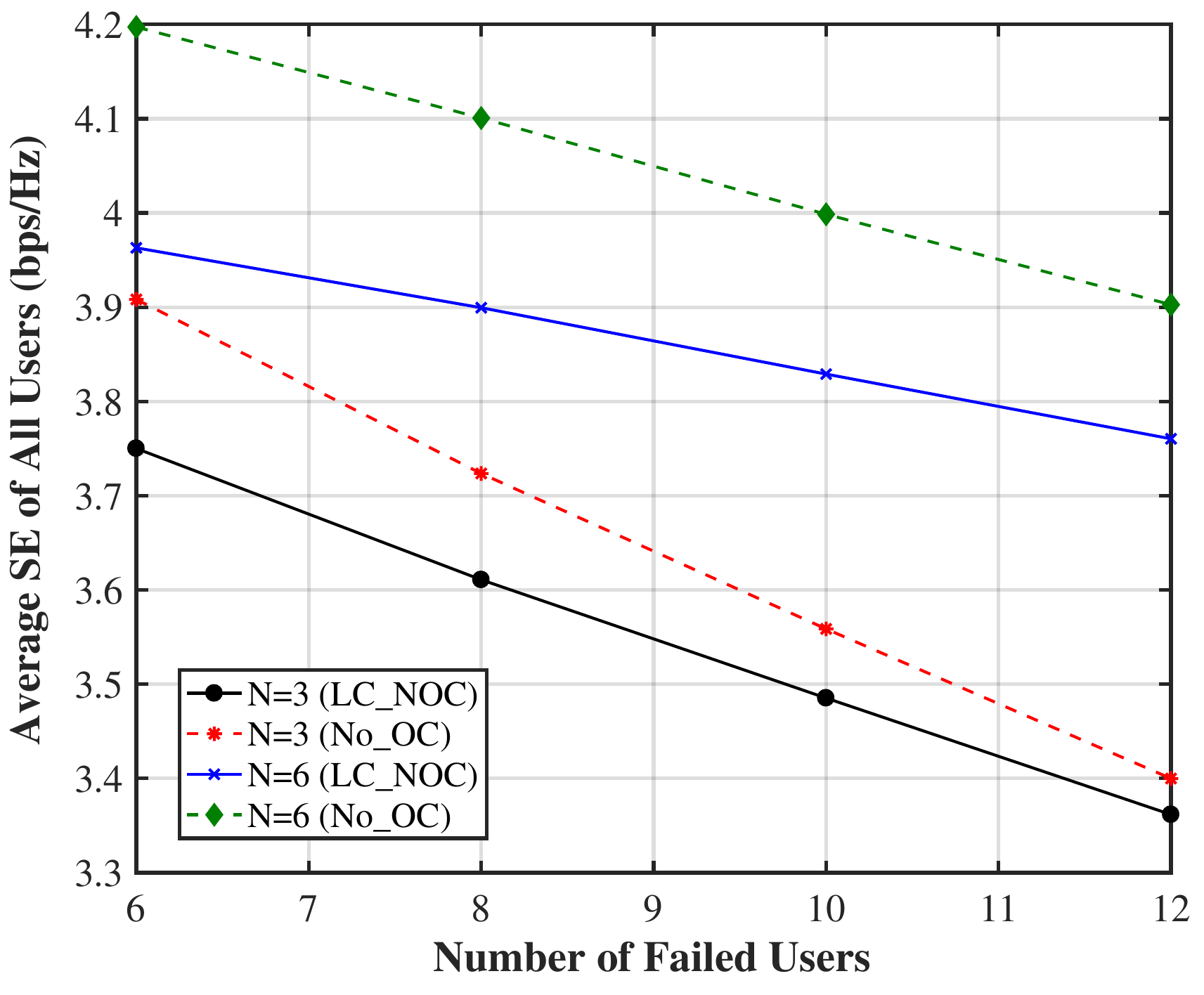}
        \label{fig:AverageSE}
    }
    \caption{Comparison of the LC\_NOC method to the No\_OC method, (a) Fairness vs. the number of failed users, (b) Average SE vs. the number of failed users.}
    \label{fig:noCompensation}
\end{figure}

\subsubsection{Effect of the number of failed users}
In this experiment, we evaluate the effect of $U^f$ and $N$ on the average SE of failed users. We also compare the results of DNN\_PA with OPT\_PA. Figure \ref{fig:N_FUE} shows that the average SE of failed users decreases with the increase in the number of failed users. This is due to the fact that when the number of failed users is large, the remaining resources are shared among a higher number of failed users after providing all connected users with their minimum data rate requirement. Comparing the results between our method and the optimal one shows that DNN\_PA can approach the average SE in proximity of the optimization-based method.

It should be noted that in case of a large number of failed users, an appropriate admission control policy is required with a view to maximizing the number of users served in the network. Admission control is beyond the scope of this paper.
 \begin{figure}[!t]
	\centering
	\centerline{\includegraphics[width=0.9\linewidth]{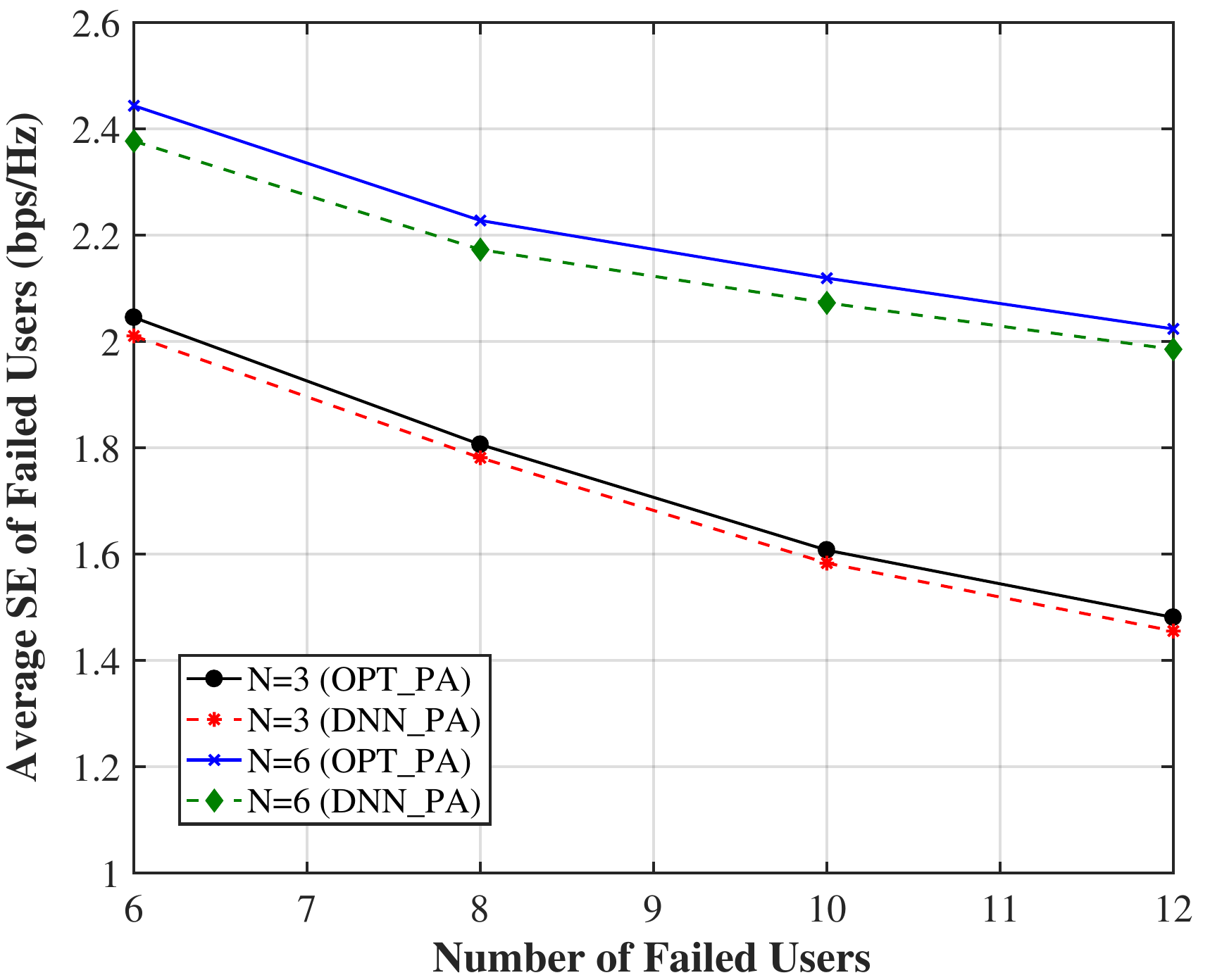}}
	\caption{Effect of the number of failed users on the average SE of failed users as a function of number of compensating base stations ($N$).}
	\label{fig:N_FUE}
\end{figure}

\subsubsection{Constraint satisfaction of the DNN-based method}
In this subsection, we conduct an experiment to investigate the satisfaction of the constraints in Problem \ref{pr:power} with DNN\_PA. We test the model on the unseen test dataset with 8,600 samples and evaluate the results. Constraints \textbf{C1}, \textbf{C2}, and \textbf{C4} are satisfied for all users.

For constraint \textbf{C3}, we calculate the relative error between the achieved SE of each connected user and its minimum SE requirement, i.e., $\mathfrak{s}^{\text{min}}_u$. The results are depicted in Figure \ref{fig:constraint}, which shows that a relative error of less than 0.01 for 98.9 percent of users whose achieved SE deviate from their minimum requirement.

 It should be noted that these constraints prevent the optimization-based methods from having a closed-form solution and result in a long execution time. On the other hand, the DNN-based method greatly reduces the online complexity, as discussed in Section \ref{sec:complexity}, with very little constraint violation. It is also worth noting that there are some methods for improving the DNN in this area, such as penalizing the constraint violation in the loss function during the training of neural network or implementing a cascading DNN with more than one neural network. However, an investigation of these methods will be left for future work.

\begin{figure}[!t]
	\centering
	\centerline{\includegraphics[width=0.9\linewidth]{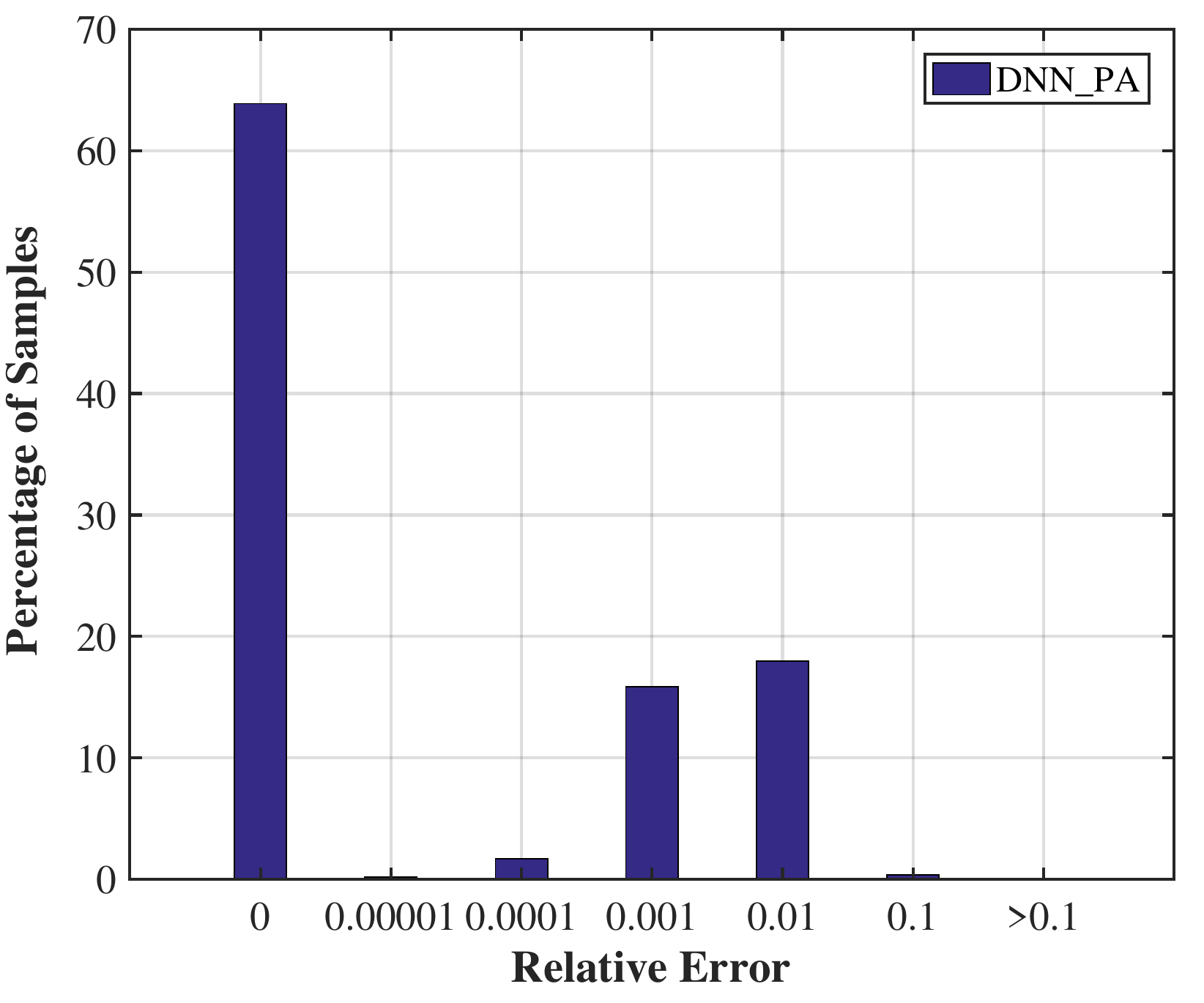}}
	\caption{Distribution of deviation from constraint \textbf{C3}.}
	\label{fig:constraint}
\end{figure}

\subsubsection{Average runtime performance}
In this experiment, we compared the average runtime of LC\_NOC and that of OPT\_NOC. We conducted our simulations on a computer with a 2.60 GHz Intel Core i7 CPU and 16 GB of RAM. We assumed that $U_n$ was 12 for each $n \in \mathcal{N}$ and $N=3$. Tables \ref{tab:Online} and \ref{tab:Offline} show the average offline and online runtime required for each of these methods, respectively. To calculate online runtime for OPT\_NOC, we first computed the average runtime of Problem \ref{pr:power}, and then multiplied it by the number of possible associations.
\\
As mentioned above, LC\_NOC needs offline dataset generation and DNN training. In general, DNN-based algorithms suffer from resource constraints in the offline phase of implementation. However, since powerful processors can be utilized for offline procedure, the runtime of this phase is not challenging. Table \ref{tab:Online} shows that LC\_NOC works efficiently in real time even when the size of the network increases, while the runtime of OPT\_NOC increases exponentially in relation to the size of the network. Therefore, proposing algorithms with high performances (close to the optimal solutions) and short runtime is of essential for real scenarios. The above simulations demonstrate that the online DNN runtime is short. Since DNN complexity is $O(N)$ and $N$ is assumed to be fixed in this experiment, the runtime remains constant.

\begin{table}[!t]
\caption{Online runtime}
\label{param}
\centering
 \begin{tabular}{|| c || c |  c ||  c ||}
 \hline
\multirow{2}{*}{ $U^f$}  &  \multicolumn{2}{ c||}{ \small LC\_NOC} & \multirow{2}{*}{ \text{\small OPT\_NOC}} \\ [0.5ex]
 \cline{2-3}
 & \small Heuristic Algorithm & \small DNN & \\  \hline
     \small 12 & \small 0.016255 s &\small 0.015625 ms & \small 2.5 $\times$ 10\textsuperscript{15} s \\  \hline
     \small 8  & \small 0.013926 s & \small 0.015625 ms & \small 1.4 $\times$ 10\textsuperscript{10} s \\  \hline
     \small 4  & \small 0.011790 s & \small 0.015625 ms & \small 5.9 $\times$ 10\textsuperscript{5} s \\  \hline
    \end{tabular}
    \label{tab:Online}
\end{table}

\begin{table}[!t]
\caption{Offline runtime}
\label{param}
\centering
 \begin{tabular}{|| c |  c ||  c ||}
 \hline
 \multicolumn{2}{|| c||}{\small LC\_NOC} & \multirow{2}{*}{\small \text{OPT\_NOC}} \\
 \cline{1-2}  \cline{1-2}
\small \text{Dataset Generation} & \small \text{Neural Network Training} &  \\ [0.5ex]
 \hline
   \small
   28.62 hours & 14.45 min & 0  \\ \hline
    \end{tabular}
    \label{tab:Offline}
\end{table}

\subsubsection{Fairness evaluation in the presence of co-channel interference}

We conducted an experiment to evaluate our proposed solution in presence of co-channel interference discussed in Subsection \ref{sec:cochannelInt}. For DNN, we use the same architecture as discussed in Subsection \ref{sec:results} except for the input layer. A dataset containing 100,000 samples is generated using an off-the-shelf solver for Problem 3 (given the user association matrix) in MATLAB. After data preprocessing, the DNN is trained with batch learning method and is tested using the test dataset.
\\
We vary the number of failed users and evaluate its effect on Jain's fairness index. In this experiment, the algorithm presented in Subsection \ref{sec:cochannelInt} is compared against the case where no compensation technique is used after outage. As can be seen from Figure \ref{fig:coInt}, the trend is the same as Figure \ref{fig:fairness}, and our algorithm shows higher fairness. However, the average SE is lower due to the co-channel interference which can be improved with an efficient channel allocation algorithm. The problem of channel allocation is beyond the scope of this paper and is left to future work.
\begin{figure}[!t]
	\centering
	\centerline{\includegraphics[width=0.9\linewidth]{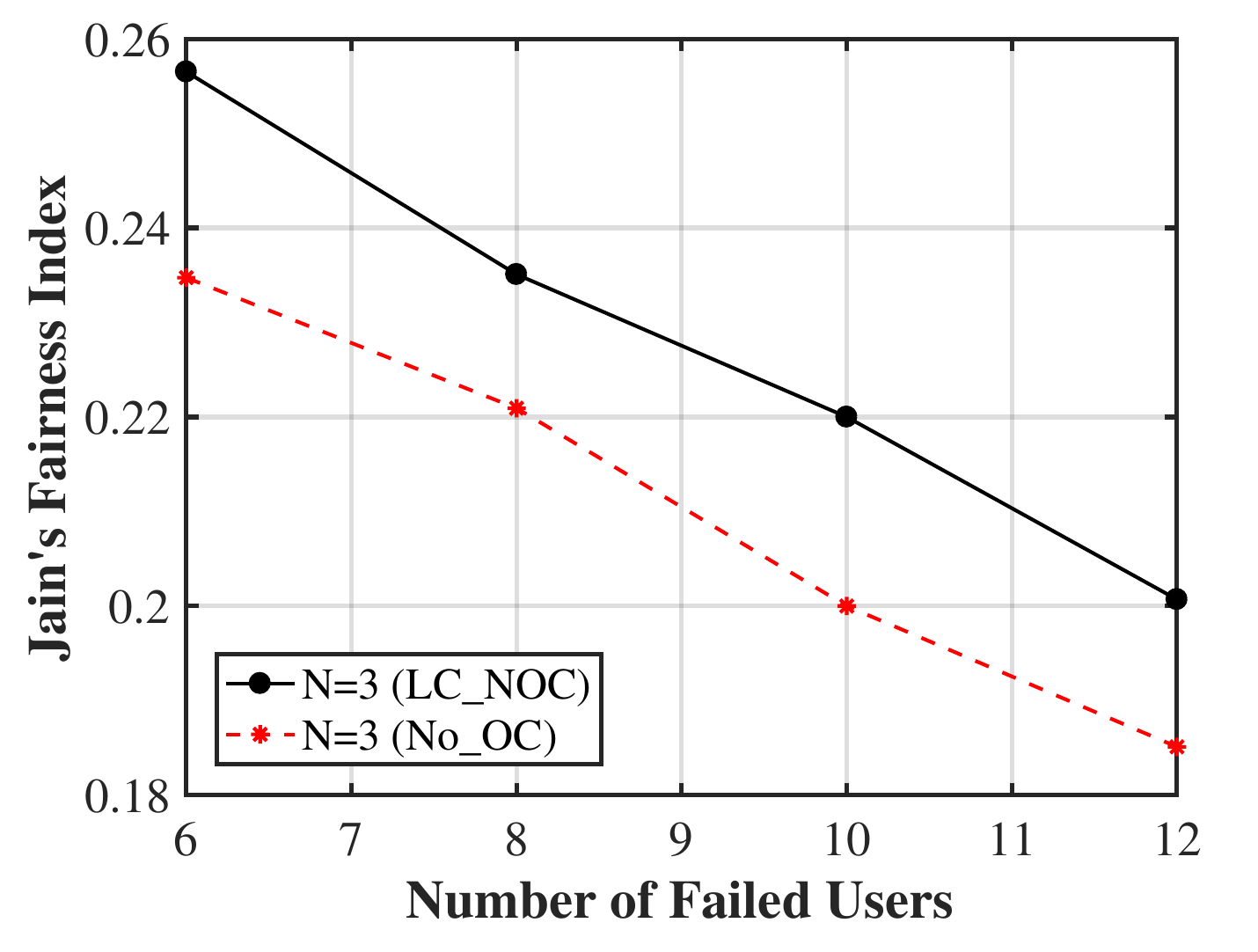}}
	\caption{Jain's fairness index evaluation in presence of co-channel interference.}
	\label{fig:coInt}
\end{figure}

\section{Conclusion}\label{sec:conclusion}
In this work, we proposed a newly scheme for cell outage compensation in NOMA-based systems. The scheme aims to serve users in the outage zone by surrounding cells in a way that maximize the SE of failed users while providing connected users  with a predefined level of QoS. We formulated the compensation process as a joint failed user association and power allocation problem that was NP-hard. An innovative, low complexity, suboptimal solution was proposed, where we associated the users in the outage cell with a heuristic algorithm and allocated their transmit power through a DNN-based approach. The proposed algorithm was shown to significantly improve the computational complexity, i.e., polynomial order with respect to the exponential complexity of finding an optimal solution. Simulation results demonstrated that the performance of the proposed algorithm approached the optimal solution. The next steps in this research could involve investigating the effect of other medium access procedures including but not limited to admission control policy on the performance of the network in cell outage compensation processes.

\bibliographystyle{IEEEtran}
\bibliography{refs}

\begin{IEEEbiography}[{\includegraphics[width=1in,height=1.25in,clip,keepaspectratio]{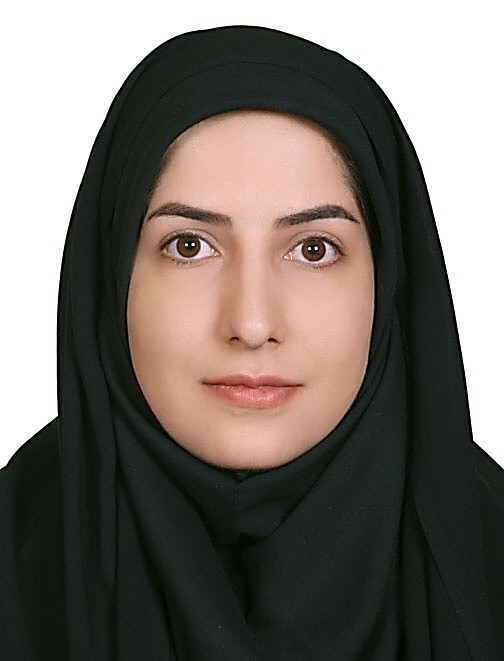}}]{Elaheh Vaezpour} received the Ph.D., M.Sc., and B.Sc. degrees in computer engineering from Amirkabir University of Technology (Tehran Polytechnic) in 2017, 2011, and 2008, respectively. She was a Visiting Scholar with the University of California, Irvine (UCI), USA and also the Newcastle University, Newcastle, Australia in 2016 and 2007, respectively. She is currently an Assistant Professor with Iran Telecommunication Research Center, Tehran, Iran. Her current research interests include radio resource allocation in wireless networks and applications of optimization theory and machine learning in system designs.
\end{IEEEbiography}

\begin{IEEEbiography}[{\includegraphics[width=1in,height=1.25in,clip,keepaspectratio]{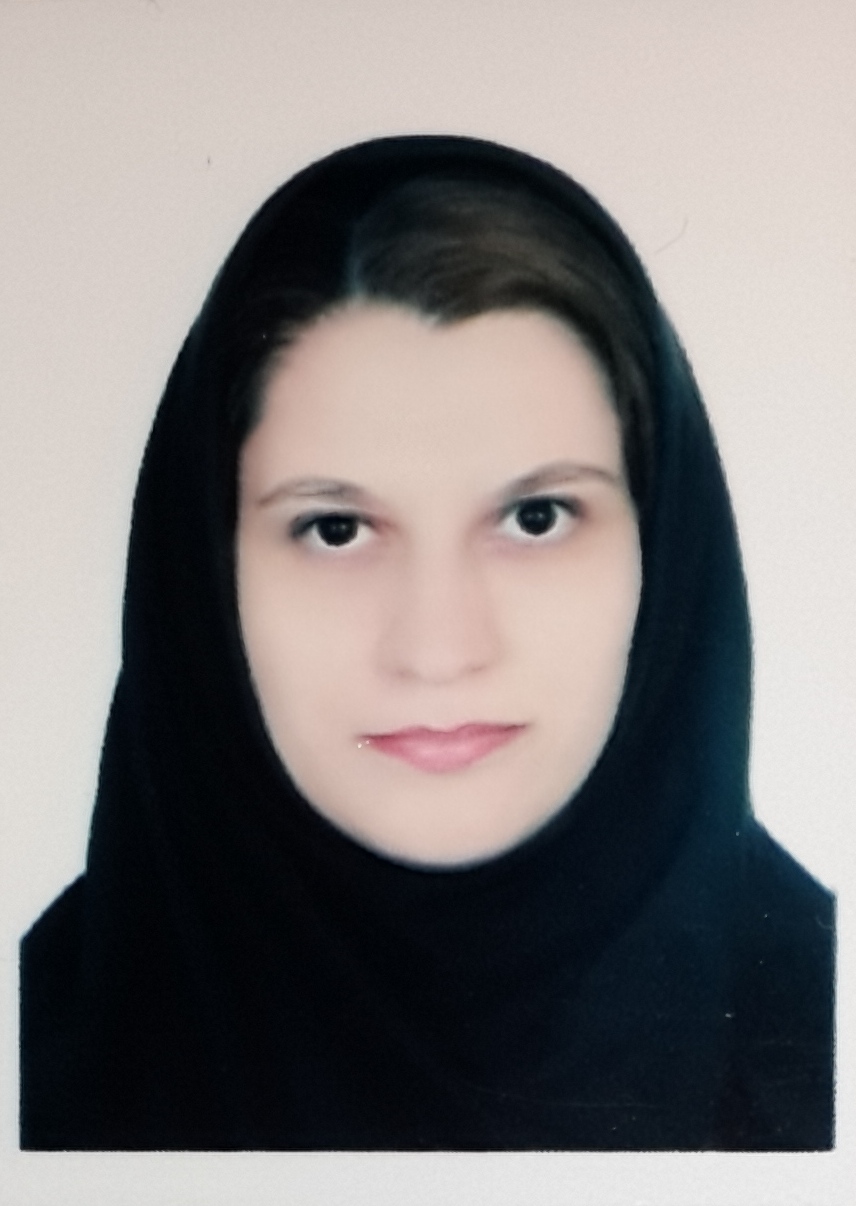}}]{Layla Majzoobi} received the B.S and M.S. degrees in electrical engineering from the Amirkabir University of Technology, Tehran, Iran in 2006 and 2009, respectively, and the Ph.D degree from the University of Tehran, Tehran, Iran, in 2019. She has been Research Assistant at the Iran Telecommunication Research Center, Tehran, Iran, from 2018. Her current research interests include large scale distributed optimization, machine learning, and 5G and 6G networks.
\end{IEEEbiography}

\begin{IEEEbiography}[{\includegraphics[width=1in,height=1.25in,clip,keepaspectratio]{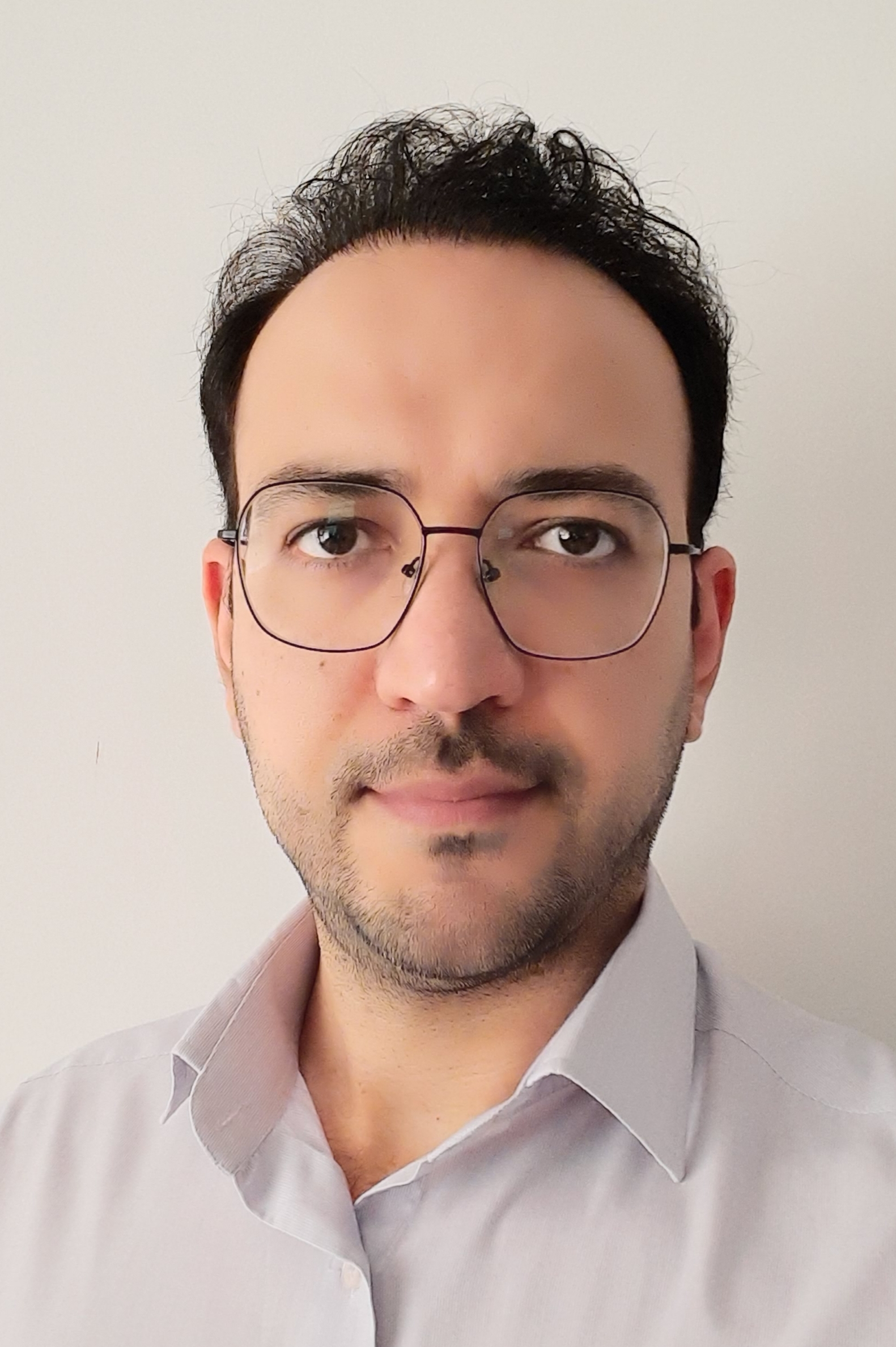}}]{Mohammad Akbari} received his B.Sc. in Electrical Engineering in 2008 from Tabriz University, Tabriz, Iran and the M.Sc. and Ph.D. degrees both from Iran University of Science and Technology (IUST), Tehran, Iran in 2010 and 2016 respectively. During 2010-2017, he was a senior system designer at Afratab R\&D group, Tehran, Iran. In 2017, he joined as a research assistant professor to the Department of Communication Technology, ICT Research Institute (ITRC), Tehran, Iran. His current research interests span topics in telecommunication system and networks including Self-Organizing Networks, 5G and 6G Networks and application of Machine Learning techniques in wireless communication.
\end{IEEEbiography}

\begin{IEEEbiography}[{\includegraphics[width=1in,height=1.25in,clip,keepaspectratio]{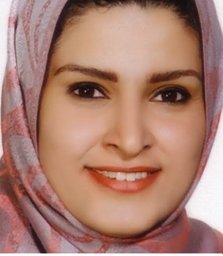}}]{Saeedeh Parsaeefard} (IEEE Senior Member) is currently a research scientist and visiting faculty member in University of Toronto. Her research has a special focus on applying optimization theory, game theory, and machine learning tools for better understanding and analyzing interactions of heterogonous, non-cooperative, distributed multi-agent systems in uncertain environments; and in particular wireless networks (5G and 6G), and IoT, IIoT, and URRLC applications. She has more than 50 journal and conference papers all in the best venue of IEEE transactions and comsoc flag conferences. Her research has been cited more than 1000 times with h-index of 16 and i10-index of 27. She is a co-author of two books in wireless network virtualization and robust resource allocation in future wireless networks, published by Springer. She received the Ph.D. degree in electrical and computer engineering from Tarbiat Modares University in 2012. From November 2010 to October 2011, she was a visiting Ph.D. student with the University of California, Los Angeles, Los Angeles, CA, USA. She was a Postdoctoral Research Fellow with the Department of Electrical and Computer Engineering, McGill University, Montreal, QC, Canada from 2013 to 2015. She received the IEEE Senior Membership and IEEE Women in Engineering awards (in region 8) in 2018.
\end{IEEEbiography}

\begin{IEEEbiography}[{\includegraphics[width=1in,height=1.25in,clip,keepaspectratio]{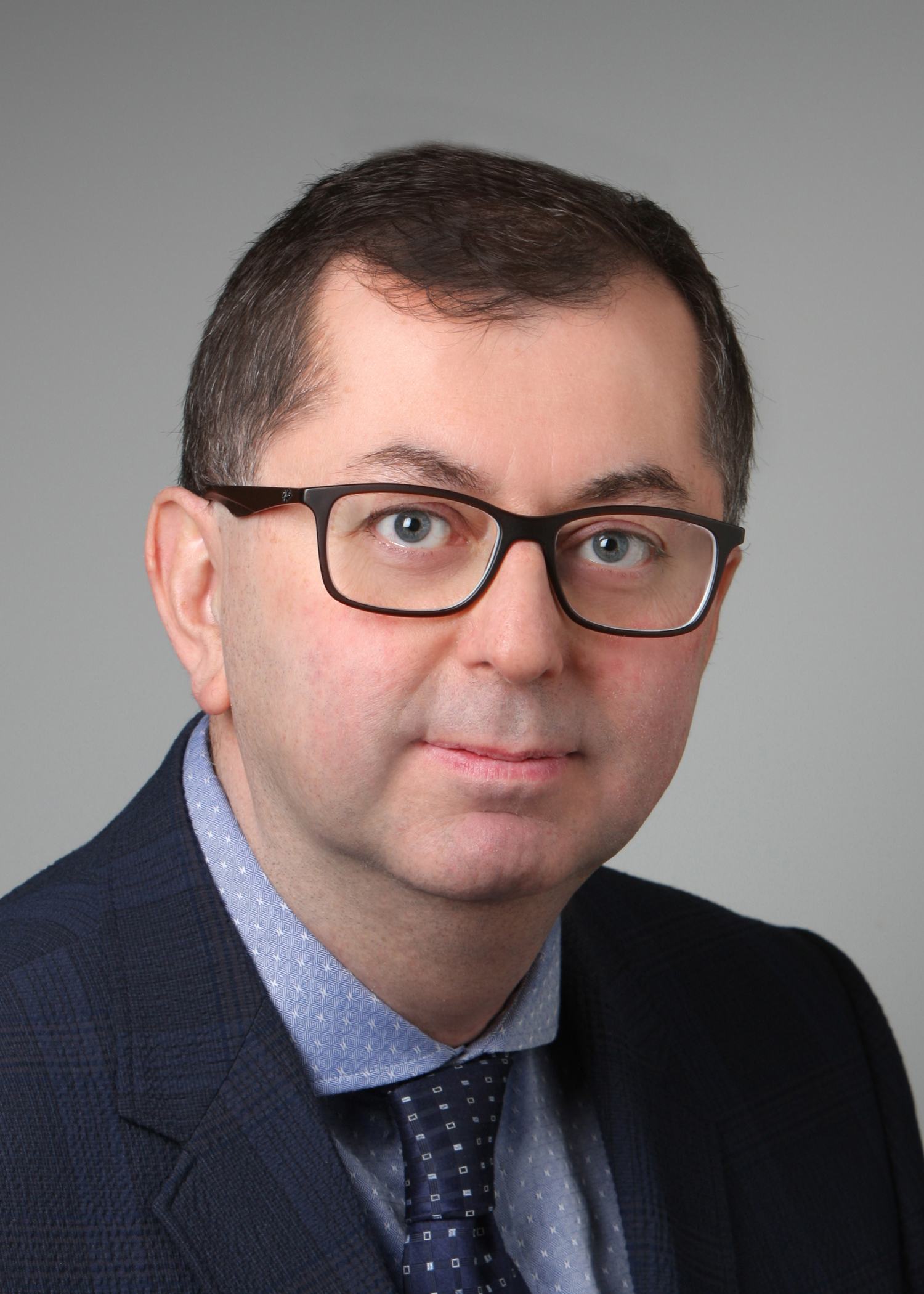}}]{Dr. Halim Yanikomeroglu} is a Professor in the Department of Systems and Computer Engineering at Carleton University, Ottawa, Canada. His primary research domain is wireless communications and networks. His research group has made contributions to 4G and 5G wireless technologies. In recent years, his work has focused on non-terrestrial networks including UAVs as users and base stations, high altitude platform stations, and dense satellite networks. His collaboration with industry has resulted in 39 granted patents. He is a Fellow of IEEE, EIC (Engineering Institute of Canada), and CAE (Canadian Academy of Engineering), and a Distinguished Speaker for both IEEE Communications Society and IEEE Vehicular Technology Society. He is currently serving as the Chair of the IEEE WCNC (Wireless Communications and Networking Conference) Steering Committee.
\end{IEEEbiography}

\end{document}